\documentclass[twoside,12pt]{paper}

\usepackage[T1]{fontenc}
\usepackage{microtype}

\usepackage[hmarginratio=1:1, left=20mm, right=20mm,top=25mm,bottom=20mm,columnsep=18pt]{geometry}

\usepackage{paralist}
\usepackage{multicol,caption}

\usepackage{graphicx}

\usepackage{lettrine}

\usepackage{amsmath,amsthm,amssymb,amscd,graphicx,latexsym,amsfonts,mathrsfs}
\usepackage{graphicx}

\usepackage{enumerate}
\usepackage{dsfont}
\usepackage[numbers,sort&compress]{natbib}

\usepackage{color}
\usepackage{chemarrow,stackrel}

\usepackage{balance}

\newtheorem{theorem}{Theorem}
\newtheorem{lemma}{Lemma}
\newtheorem{definition}{Definition}
\newtheorem{assumption}{Assumption}
\newtheorem{remark}{Remark}

\def\EE{{\mathbb E}}
\def\NN{{\mathbb N}}
\def\PP{{\mathbb P}}
\def\RR{{\mathbb R}}

\def\ONE{{\mathds{1}}}

\def\d{{\mathrm d}}

\def\ll{{\mathscr L}}
\def\tll{{\tilde{\mathscr L}}}
\def\mm{ { \mathcal M} }

\def\al{\alpha}
\def\bt{\beta}

\def\Del{\Delta}

\def\lmd{\lambda}

\def\Del{\Delta}

\def\bA{\bar{A}}
\def\ba{\bar{a}}

\def\bal{\bar{\alpha}}
\def\bB{\bar{B}}

\def\bdel{\bar{\delta}}
\def\bF{\bar{F}}

\def\bP{ \bar{P} }

\def\hDelt{\widehat{\Delta t}}

\def\oDt{\overline{\Delta t}}
\def\tDel{\tilde{\Delta}}
\def\tDelt{\widetilde{\Delta t}}

\def\uDt{\underline{\Delta t}}

\def\tV{\tilde{V}}

\usepackage{fancyhdr}
\title{\vspace{-15mm}%
	\fontsize{24pt}{10pt}\selectfont
	\textbf{
 Stability of cyber-physical systems of numerical methods for stochastic  differential equations: integrating the cyber and the physical of stochastic systems }
	}	
\author{ 
Lirong Huang  
\thanks{Guangzhou, China  (email:lrhuang@aliyun.com). The author was with School of Automation, Guangdong University of Technology, Guangzhou, China}
}

\date{}

\balance

\begin{document}

\maketitle
\thispagestyle{fancy}

\begin{abstract}

This paper presents the cyber-physical system (CPS) of a numerical method (the widely-used Euler-Maruyama method) and establishes a foundational theory of  the CPSs of numerical methods for stochastic differential equations (SDEs), which transforms the way we understand the relationship between the numerical method and the underlying dynamical system.  Unlike in the literature where they are treated as separate systems linked by inequalities, the CPS is a seamless integration of the SDE and the numerical method and we construct a new and general class of stochastic impulsive differential equations (SiDEs) that can serve as a canonic form of the CPSs of numerical methods. By the CPS approach, we show the equivalence and intrinsic relationship between the stability of the SDE and the stability of the numerical method using the Lyapunov stability theory we develop for our class of  SiDEs. Applying our established theory, we present the CPS Lyapunov inequality that is the necessary and sufficient condition for mean-square  stability of the CPS of the   Euler-Maruyama method for linear SDEs.
The proposed CPS and theory initiate the study of systems numerics and provoke many open and interesting problems for future work.

\end{abstract}

\section{Introduction}        \label{sec:introduction}

According to Newton's second law of motion, we  describe a mechanical system with differential equations. Usually, physiccal laws are expressed by means of differential equations, and so are the models of dynamical systems in many disciplines,  ranging from biology to finance. Such models play a central role in all scientific and engineering disciplines \cite{control2002,lee2015}. A model may serve many purposes. The value of a model depends on the  model fidelity,  where the model of a dynamical system is said to be have high fidelity if it accurately describes the properties of the system \cite{control2002,lee2015,lee2017}. Studying the model of high fidelity gives us insight into how the dynamical system will behave in the real world.  Generally, a dynamical system, ranging from the motion of  pollen particles  to the movement of  stock price, is subject to intrinsic and/or extrinsic noise in the real world  \cite{higham2008,huang2010PhD,mao2007book}.  Such randomness should/must be taken into account  by a model of high fidelity if it matters,  say, it affects   some  property of the  system that is of concern to the modelling.   If we allow for some noise in some coefficients of a differential equation, we often obtain a more realistic model of the situation that is able to describe  the fluctuations observed in the physical system.  This leads   to modelling with  stochastic differential equations (SDEs). The study of SDEs can be seen to have started from the classical paper
of Einstein that presented a mathematical connection between
microscopic Brownian motion of particles and the macroscopic diffusion
equation, and the interest in SDE has grown enormously in the last few decades \cite{arnold1974,khas2012,oksendal2000,sarkka2019,vankampen2007,wilkinson2012}.

Over the past a few decaseds, stochastic systems described by SDE have been intensively studied since stochastic modelling has come to play a significant role  in many branches of science and engineering \cite{astrom1970,huang2009,huang2010,huang2016,huang2022}. It is  hardly possible to solve an SDE analytically and have the exact solution of the SDE.  For practical purposes, numerical approximations to the exact solution are usually obtained, which, called the numerical solutions,  are  discrete-time stochastic processes produced by some numerical methods. Such numerical schemes, in the form of stochastic difference equations, are  the translations of the SDE into discretization.   Practically, computers are used to excute the numerical methods and generate the numerical solutions of the SDE,   from   which one could learn and/or infer some dynamical properties of the underlying  physical system \cite{astrom2014,higham2001,huang2016,stuart1996}. 

As is well known, whenever a computer is used
in measurement, computation, signal processing or control applications, the
data, signals and systems involved are naturally described as discrete-time
processes \cite{astrom1997,huang2015,huang2016,scotton2013}.  It is worth noting that the SDE is the physical model which represents our knowledge of the physical system and a numerical method is a cyber model which is a  representative of the physical model in  computers, the cyber world.  The physical model, namely, the SDE often refers to the phyiscal system (particularly,  which is an engineered system) itself  while its cyber couterpart, namely, the numerical method symbolizes it in the cyber world. In the age of networking and information technology, the cyber model  plays a key role in understanding and controlling the underlying physical system, which not only envisions the approximate behaviour of the physical system \cite{higham2001,higham2002,higham2003} but is also   utilized to extract knowledge of the system from data \cite{huang2016,scotton2013} and   based on which control is designed and implemented \cite{astrom1970,astrom1997,heirung2018,huang2015,huang_partII,nghiem2006}.
It is natural and imperative
\begin{itemize}
 \item[(I)]  to  find out the relationship between the physical model (i.e., the SDE)  and its cyber counterpart (i.e., the numerical method) of a dynamical system;

 \item[(II)] and to ensure that they both share some dynamical properties such as stability which is   concerned  in the study. 
\end{itemize}

The principal aim of this  paper is to address the  problems (I) and (II) of fundamental importance in the age of networking and information technology. 
As a matter of fact, the fundamental importance of these problems has been recognized and  they have been studied in a vast literature. Results that address the  problems can be found in those many on convergence and stability of numerical methods for SDEs, where  the SDE and the numerical method are treated as separate systems which  are linked by inequalities in some moment sense on any finite time interval \cite{higham2002, higham2003,hutzenthaler2012, kloeden1992,mao2015b, sabanis2013}. The ability of a cyber system to reproduce the stability of its underlying physical system can be found in a large number of works.  The problem how to reproduce the stability of the physical system in its cyber counterpart, which is called the test problem, has also been  studied for  SDE   \cite{higham2000, higham2003, mao2015a, saito1996}. The key question  in a test problem is  \cite{higham2000}

\begin{itemize}
 \item[(Q1)] for what stepsizes $\Delta t$ does the cyber system (the numerical method)  share the stability property of the underlying physical system (the SDE)?
\end{itemize}
\noindent  This  naturally provokes the converse question     \cite{higham2003,mao2015a}
\begin{itemize}
 \item[(Q2)]  does the stability of the cyber system (the numerical method)   for  small stepsizes $\Delta t$ imply that of the underlying physical system (the SDE)?
\end{itemize}

These questions deal with asymptotic ($t \to \infty$) properties and hence cannot be answered directly by applying traditional finite-time convergence results  \cite{higham2003, mao2015a}. Results that answer questions (Q1) and (Q2)  can be found in the literature. For example, results for scalar linear systems were given in  \cite{higham2000, saito1996}. For multi-dimensional systems with global Lipschitz condition,
Higham, Mao and Stuart \cite{higham2003} introduced a natural finite-time strong convergence condition \cite[Condition 2.3]{higham2003}, which links a cyber system with its underlying physical system  by an inequality  in some moment sense over any finite time interval,    and proved that there is a sufficiently small $\Delta t^\ast >0$ such that, for every  $\Delta t \in (0, \Delta t^\ast]$,  the mean-square exponential stability of  the physical system  is equivalent to that of its cyber  counterpart   \cite[Lemmas 2.4-2.5 and  Theorem 2.6]{higham2003}.  Recently, Mao \cite{mao2015a} developed new techniques to handle the small $p$th moment ($p \in (0,1) $) and showed that, under a natural finite $p$th moment condition \cite[Assumption 2.4]{mao2015a} and a natural finite-time convergence condition  \cite[Assumption 2.5]{mao2015a},   the $p$th moment exponential stability of  the physical system  is equivalent to that of its cyber  counterpart  for every  $\Delta t \in (0, \Delta t^\ast]$ with some sufficiently small $ \Delta t^\ast >0$ \cite[Lemmas 2.6-2.7 and Theorem 2.8]{mao2015a}.  As is pointed out in \cite{mao2015a}, there are many open problems in this research. For instance, although the existence of the (sufficiently small) upper bound $\Delta t^\ast >0$ of stepsizes   has been shown \cite{higham2003,mao2015a}, it is severely limited by the growth constant of the exponentially stable system, which refers to the physical system  and its cyber counterpart  when answering  (Q1) and (Q2), respectively.  Recall that, though the growth  and the rate constants are related, it is  only the rate constant  that counts in the definition of exponential stability, see Definitions \ref{pthM_ExpStab} and \ref{AS_ExpStab} below. It appears that, either to reproduce or to imply  the exponential stability of the physical system  by  its cyber counterpart, the condition imposed on the  stepsizes that explicitly depends on the growth constant could/should be relaxed \cite{stuart1996}. This could significantly improve the upper bound $\Delta t^\ast $ of stepsizes and facilitate the computation.

It is noted that the physical system (the SDE) and its   cyber counterpart (the numerical method) are bound by inequalities in the literature \cite{higham2002, higham2003,hutzenthaler2012, kloeden1992,mao2015b, sabanis2013,saito1996,stuart1996}.  Nevertheless,   they  remain as two systems, largely  separate. 
This paper constructs the cyber-physical model of a dynamical system that is a seamless, fully synergistic integration of the physical model  (the SDE) and its cyber counterpart (the numerical method). Here we present a new and general class of stochastic impulsive differential equations (SiDEs) which can be used to represent the integrated dynamics of the physical system and its cyber counterpart. Impulsive differential equations, also known as impulsive systems,  have been studied  for several decades  \cite{caraballo2017,hespanha2008,huang2022,peng2010,samoilenko1995,yang2001}.  But these impulsive systems in the literature are just the physical subsystems in our  general class of SiDEs, see Section \ref{sec:generalSiDEs}. Our proposed SiDE  composed of the physical subsystem and the cyber subsystem is formulated as a canonic form of cyber-physical systems (CPSs), which presents a systematic framework for  the study of CPSs \cite{cps2008,koutsoukos2014}. The  canonic form  not only provides a holistic view of CPSs but also reveals the intrinsic relationship between the physical subsystem  and the cyber subsystem. 
In the study of numerical analysis, we present the CPS of a numerical method (the widely-used Euler-Maruyama  method) for SDEs, which  is in the form of our proposed SiDEs and represents  a seamless integration of the SDE and the numerical method. The SDE and the numerical method are the physical subsystem and the cyber subsystem of the CPS, respectively.  
  From the viewpoint of cybernetics \cite{wiener1961}, an essential problem to study is whether and how the CPS  reproduces some dynamical properties such as the stability of its physical or cyber subsystem  since `the primary concern of cybernetics is on the qualitative aspects of the interrelations among the various components of a system and the synthetic behavior of the complete mechanism'  \cite{tsien2015}.  Using the terminology of the CPSs, we rephrase the  key questions (Q1) and  (Q2)   as follows.

(Q1) For what stepsizes $\Delta t$ do the CPS and, hence,  the cyber subsystem    reproduce the stability property of the physical subsystem?

(Q2) Does the stability of the cyber subsystem for  small stepsizes $\Delta t$ imply that of the CPS  and, hence, that of the physical subsystem?

To address the key questions, we develop a Lyapunov stability theory for our proposed general class of SiDEs and apply it to the CPS of the  Euler-Maruyama  method for SDEs. Applying our established stability theory, we  prove positive results to the two key questions (Q1) and (Q2). This establishes a foundational theory for the CPSs of   numerical methods for SDEs. As  application of our established theory, we present the CPS Lyapunov inequality,   the necessary and sufficient condition for mean-square exponential stability of the CPS of  the  Euler-Maruyama  method for linear SDEs. In this paper, we initiate the study of systems numerics and there are many interesting and/or challenging problems for future work.

\section{A general class of stochastic impulsive differential equations}  \label{sec:generalSiDEs}

Throughout  this   paper, unless otherwise specified, we shall employ the
following notation.  Let us denote by  $( \Omega , \mathcal{F}, \{ \mathcal{F}_t
\}_{t \ge 0}, \PP ) $ a complete probability space with a
filtration $ \{ \mathcal{F}_t \}_{t \ge 0} $ satisfying the usual
conditions (i.e. it is right continuous and $\mathcal{F}_0$ contains
all $\PP$-null sets) and by ${\mathbb{E}}[\cdot]$   the expectation
operator with respect to the probability measure. Let $B(t)= \begin{bmatrix}
B_1(t) & \cdots & B_m(t) \end{bmatrix}^T$ be an $m$-dimensional Brownian motion
defined on the probability space. If $x, y$ are real numbers, then
$x \vee y$ denotes the maximum of $x$ and $y$, and $x \wedge y$
stands for the minimum of $x$ and $y$. If $A$ is a vector or a matrix, its transpose is denoted by $A^T$. If $P$ is a square matrix, $P>0$ (resp. $P<0$) means that P is a symmetric positive (resp. negative) definite matrix of appropriate dimensions 
while $ P \ge 0$ (resp. $P \le 0$) is a symmetric positive (resp. negative) semidefinite matrix. Let  $\lambda_M( \cdot) $ and $\lambda_m( \cdot)$ be  a matrix's  eigenvalues   with
maximum  and  minimum real parts, respectively. 
 Denote by $| \cdot |$  the
Euclidean norm of a vector and the trace (or Frobenius) norm of a matrix. 
%

Let $C^{2, 1} (\RR^n \times \RR_+; \RR_+)$ be the family of all nonnegative functions $V(x, t)$ on $\RR^n \times \RR_+$ that are continuously twice differentiable in $x$ and once in $t$. Let $\mm^p ([a, b]; \RR^n)$ be the family of $\RR^n$-valued adapted process $\{ x(t): a \le t \le b \}$ such that $\EE \int_a^b |x(t)|^p \d t < \infty$. 
 Denote by $I_n$ the $n \times n$ identity matrix and  by $0_{n \times m}$ the $n \times m$  zero matrix, or, simply, by $0$  the zero matrix of compatible dimensions. Let $\{ \xi (k) \}_{k \in \NN} $, $\NN =\{ 0, 1, 2, \cdots \}$, be an independent and identically distributed sequence with $\xi (k)= \begin{bmatrix}
\xi_1(k) & \cdots & \xi_m(k) \end{bmatrix}^T$, and $ \xi_j(k) $, $j =1, 2, \cdots, m$, obeying standard normal distribution while $\{t_k\}_{k \in \NN}$ is a strictly increasing sequence which satisfies $0=t_0< t_1 < t_2 < \cdots$, $0<\underline{\Delta t} :=\inf_{k \in \NN} \{ t_{k+1} -t_k \} \le \overline{\Delta t} :=\sup_{k \in \NN} \{ t_{k+1} -t_k \} < \infty$, and hence $t_k \to \infty$ as $k \to \infty$. 

Let us consider  a stochastic impulsive system described by SiDE 
\begin{subequations}  \label{SiDE-xy}
\begin{align}
& \mathrm{d}  x(t) = f(x(t), t)\mathrm{d}t + g(x(t),t )\mathrm{d} B(t)   \label{SDE_x} \\
& \d y(t)  = \tilde{f} (x(t), y(t), t) \d t + \tilde{g}( x(t), y(t), t) \d B (t),   \quad  t \in [ t_k,  t_{k+1})     \label{SDE_y}     \\
& \Delta (x(t_{k+1}^-), k+1) := x(t_{k+1}) -x(t_{k+1}^-)  \nonumber  \\
& \; \;  {}   = h_f (x(t_{k+1}^-), k+1) + h_g  (x(t_{k+1}^-), k+1)   \xi (k+1)  \label{Impulse_x} \\
&\tilde{\Delta} (x(t_{k+1}^-), y(t_{k+1}^-), k+1) := y(t_{k+1}) - y(t_{k+1}^-)  \nonumber  \\
 & \; \; {}       = \tilde{h}_f (x(t_{k+1}^-), y(t_{k+1}^-), k+1)    + \tilde{h}_g (x(t_{k+1}^-), y(t_{k+1}^-), k+1)  \xi (k+1),  \quad k \in \NN  \label{Impulse_y} 
\end{align}
\end{subequations}
 with initial data $x(0)  \in \RR^n$ and $ y(0)  \in \RR^q$, 
where $ \xi (k+1)$ is independent of $\{ x(t), y(t):  0 \le t < t_{k+1}  \}$, and  $ f : \RR^n \times \RR_+ \to \RR^n $, $g : \RR^n \times \RR_+ \to \RR^{n \times m} $,   $h_f : \RR^n \times \NN \to \RR^n $,  $h_g : \RR^n \times \NN  \to \RR^{n \times m} $, 
 $\tilde{f} : \RR^n \times \RR^q \times \RR_+ \to \RR^q $, $\tilde{g} : \RR^n \times \RR^q \times \RR_+ \to \RR^{q \times m} $, 
$\tilde{h}_f : \RR^n \times \RR^q \times \NN \to \RR^q $, $\tilde{h}_g: \RR^n \times \RR^q \times \NN \to \RR^{q \times m} $ are measurable functions. To study stability of the system, we assume that they obey
$f (0, t) =0 $, $g(0, t)=0$,      $h_f(0, k) =0$ and $h_g(0, k) =0$, 
$\tilde{f} (0, 0, t) =0 $, $\tilde{g}(0,0, t)=0$, $\tilde{h}_f(0,0, k) =0$ and $\tilde{h}_g(0,0, k) =0$ for all $t \in \RR_+$ and $k \in \NN$,
and  satisfy the global Lipschitz   conditions.
\begin{assumption}   \label{globalLipschitz}
There is $L >0$ such that 
\begin{equation*}  
   | f(x, t) - f(\bar{x}, t) | \vee |g(x, t)- g(\bar{x}, t)|   
   \vee | h_f(x, k) - h_f(\bar{x}, k) |  \vee | h_g(x, k)- h_g(\bar{x}, k)| \le L  |x- \bar{x}|    
\end{equation*} 
for all $(x, \bar{x}) \in \RR^n \times \RR^n$, $t \in \RR_+$, $k \in \NN$;  and there is $\tilde{L}>0$ such that
\begin{multline*} 
   | \tilde{f}(x, y, t) - \tilde{f}(\tilde{x}, \tilde{y}, t) | \vee | \tilde{g}(x, y, t) - \tilde{g}(\tilde{x}, \tilde{y}, t) | \\
    \vee | \tilde{h}_f(x, y, k) - \tilde{h}_f(\tilde{x}, \tilde{y}, k) |  \vee  | \tilde{h}_g(x, y, k) - \tilde{h}_g(\tilde{x}, \tilde{y}, k) |    \le \tilde{L} ( |x- \tilde{x}|  \vee | y - \tilde{y}|)    
\end{multline*}
for all $(x, y, \tilde{x}, \tilde{y}) \in \RR^n \times \RR^q \times  \RR^n \times \RR^q$, $t\in \RR_+$, $k\in \NN$.
\end{assumption}

It is the intersection, interaction and interrelation of the physical system and its cyber counterpart  \cite{astrom1997,control2002,cps2008,higham2003,huang2016,lee2010,nghiem2006,scotton2013} in the age of networking and information technology that motivate our study of stochastic impulsive system (\ref{SiDE-xy}) , which is constructed as  a canonic form of CPSs  that is a seamless, fully synergistic integration of the physical system  and its cyber counterpart, see also \cite{huang_partII}. It is observed in Section \ref{sec:numericalCPS} that the CPS of a numerical method for SDEs is a special case of (\ref{SiDE-xy}) where there is no impulse imposed on the physical subsystem (\ref{SDE_x}). 
We stress that we delibrately include the impulse (\ref{Impulse_x}) and use the impulsive subsystem  (\ref{SDE_x},\ref{Impulse_x})  in  SiDE (\ref{SiDE-xy}) to show that  the impulsive systems in the literature \cite{caraballo2017,hespanha2008,huang2022,peng2010,samoilenko1995,yang2001} are just the physical subsystems  (\ref{SDE_x},\ref{Impulse_x}) in our general class of SiDEs.  We  construct a system integration (\ref{SiDE-xy})  of two impulsive subsystems   in marked contrast to  the impulsive systems in the literature, which highlights the distinction between our new class  and those existing.

\begin{remark}  \label{remark-impulsive}
 It should be   noticed that, usually, the impulse interval of the subsystem  $x(t)$ is  a multiple of that of the subsystem  $y(t)$ since the former is actually the interval between two consecutive physical impulses imposed on  the physical system  $x(t)$  while  the latter  the stepsize of the numerical method, see Section \ref{sec:numericalCPS}. In such a specific    case of SiDE (\ref{SiDE-xy}) where $t_k = k \Delta t$,    $ h_f (x, k) $ and $h_g (x, k)$  are some functions of $x$ only if $k$ is a multiple of  integer $k_0>1$, and, otherwise,  $ h_f (\cdot, k) =0$ and $h_g (\cdot, k)=0$, where  $k_0$ is the number of steps of the numerical scheme between two consecutive physical impulses.
\end{remark}

Clearly, the trivial  solution is  the equilibrium of  system (\ref{SiDE-xy}).  
For a function $V \in C^{2,1} (\RR^n \times \RR_+; \RR_+)$,  the infinitesimal generator $\ll V: \RR^n \times \RR_+ \to \RR$ associated with system (\ref{SDE_x}) is defined as
\begin{equation}    
  \ll V (x, t) = V_t (x, t) + V_x (x, t) f(x, t)     + \frac{1}{2}   {\rm trace} \left[ g^T (x, t) V_{xx} (x, t) g(x, t) \right],  \label{LV}
\end{equation}
where 
$ V_t (x, t) =  \frac{ \partial V(x, t)} { \partial t} $, $V_x (x, t) = \begin{bmatrix}  \frac{ \partial V(x, t)}{ \partial x_1} & \cdots & \frac{ \partial V(x, t)}{ \partial x_n}  \end{bmatrix} $ and $   V_{xx} (x, t) =  \begin{bmatrix}  \frac{ \partial^2 V(x, t)}{ \partial x_i \, \partial x_j}   \end{bmatrix}_{n \times n} $.
Similarly, for a function $\tilde{V} \in C^{2,1} (\RR^q \times \RR_+; \RR_+)$, the generator  $\tll \tilde{V}:  \RR^n \times \RR^q \times \RR_+ \to \RR$ associated with system (\ref{SDE_y}) is defined as
\begin{equation}    
  \tll \tilde{V} (x, y, t) = \tilde{V}_t (y, t) + \tilde{V}_y (y, t) \tilde{f}(x, y, t)   + \frac{1}{2} {\rm trace} \left[ \tilde{g}^T (x, y,  t) \tilde{V}_{yy} (y, t) \tilde{g}(x, y, t) \right].    \label{tLtV}
\end{equation}

Let $z(t) = [ x^T(t) \; \,  y^T(t) ]^T \in \RR^{n+q}$,  $C = [  I_n  \; \,  0_{n \times q}  ]$, 
$ D =  [  0_{ q \times n}  \; \, I_{q}  ]$, and, therefore, $x(t) = C z(t)$,  $y(t) = D z(t)$ for all $t \ge 0$.
The stochastic impulsive system (\ref{SiDE-xy}) can be written in a compact form  
\begin{subequations}    \label{Compact-z}
\begin{align}
  &  \d z(t) = F ( z(t), t) \d t + G (z(t), t) \d B(t)  \quad  t \in [ t_k, t_{k+1})     \label{SDE_z} \\
  &  z (t_{k+1}) - z (t_{k+1}^-)     = H_F (z(t_{k+1}^-), k+1) + H_G  (z(t_{k+1}^-), k+1) \xi (k+1)  \quad  k \in \NN \label{Impulse_z}
\end{align}
\end{subequations}
  with initial value $z(0)  =[ x(0)^T  \; \,  y(0)^T ]^T \in \RR^{ n + q}$, where functions $F: \RR^{n+q} \times \RR_+ \to \RR^{n+q}$,  $G: \RR^{n+q} \times \RR_+ \to \RR^{(n+q) \times m}$, $H_F: \RR^{n+q} \times \NN \to \RR^{n+q}$, $H_G: \RR^{n+q} \times \NN \to \RR^{(n+q) \times m}$ are given by
\begin{multline*}
   F(z, t) = \begin{bmatrix}  f\left( C z, t \right)  \\
                      \tilde{f} \left( Cz  , D z, t \right) \end{bmatrix},
 G(z, t) =\begin{bmatrix}  g\left( C z, t \right)  \\
                      \tilde{g} \left( Cz, D z \right), t \end{bmatrix},    \\
  H_F (z, k) = \begin{bmatrix}   h_f \left( C z, k \right)  \\
                      \tilde{h}_f \left( Cz, Dz, k \right) \end{bmatrix},
 H_G (z, k) = \begin{bmatrix}    h_g \left( C z, k \right)  \\
                      \tilde{h}_g \left( C z, D z , k\right) \end{bmatrix} .
\end{multline*}
Consequently,  stochastic impulsive system (\ref{Compact-z}) obeys $F(0, t)=0 $, $G(0, t)=0$, $H_F (0, k) =0$, $H_G(0, k)=0$ for all $t \in \RR_+$, $k \in \NN$, and, by  Assumption \ref{globalLipschitz}, it satisfies the global Lipschitz  condition, that is,  there is a constant $L_z >0$ such that
\begin{multline} \label{globalLipschitz-z}
  | F(z, t) - F(\tilde{z}, t) | \vee | G(z, t) - G(\tilde{z}, t) |  \\
 \vee | H_F(z, k) - H_F(\tilde{z}, k) |   \vee | H_G(z, k) - H_G(\tilde{z}, k)) |   
       \le L_z  |z- \tilde{z}|
\end{multline}
 for all $(z, \tilde{z}) \in \RR^{n+q} \times \RR^{n+q}$, $t \in \RR_+$, $k \in \NN$. 
It is easy to obtain the following result on existence and uniqueness of solutions for SiDE  (\ref{Compact-z}).
\begin{lemma}  \label{existence_n_uniqueness}
     Under Assumption   \ref{globalLipschitz}, there exists a unique (right-continuous) solution $z(t)$ to SiDE (\ref{Compact-z}) (viz.,  (\ref{SiDE-xy})) on $t \ge 0$ and the solution belongs to  $\mm^2 ([0, T]; \RR^{n+q})$ for all $T \ge 0$.
\end{lemma}
\begin{proof}  Since system (\ref{Compact-z}) satisfies the global Lipschitz continuous condition (\ref{globalLipschitz-z}), according to \cite[Theorem 3.1, p51]{mao2007book}, there exists a unique solution $z(t)$  to SiDE (\ref{Compact-z}) on $t \in [t_0, t_1)$ and the solution belongs to $\mm^2 ([t_0, t_1); \RR^{n+q})$. Note that $\xi (k+1)$ is independent of $\{ z(t): t  \in [t_0, t_1) \}$. By virtue of continuity of functions $H_F (z, \cdot)$ and $H_G  (z, \cdot)$ with respect to $z$, there exists a unique solution $z(t_1)$ to  (\ref{Compact-z}) on $t = t_1$. Moreover, (\ref{Impulse_z}) and (\ref{globalLipschitz-z}) imply that the second moment of $z(t_1)$ is finite. And, again, according to \cite[Theorem 3.1, p51]{mao2007book}, one has that there is a unique right-continuous solution $z(t)$  to  (\ref{Compact-z}) on $t \in [t_0, t_2)$ and the solution belongs to $\mm^2 ([t_0, t]; \RR^{n+q})$ for $t \in [t_0, t_2)$. Recall that $0=t_0< t_1< t_2 < \cdots< t_k < \cdots$ and $t_k \to \infty$ as $k \to \infty$. By induction, one has that there exists a unique (right-continuous) solution $z(t)$ to SiDE (\ref{Compact-z}) for all $t \ge 0$ and the solution belongs to  $\mm^2 ([0, T]; \RR^{n+q})$ for all $T \ge 0$. 
\end{proof} 

Now that we have shown the existence and uniqueness of solutions to SiDE (\ref{Compact-z}), we shall further study the stability of the solutions to the SiDE. Let us introduce the definitions of exponential stability for  SiDE (\ref{Compact-z}).

\begin{definition}   \label{pthM_ExpStab}
 \cite[Definition 4.1, p127]{mao2007book} System (\ref{Compact-z}) is said to be $p$th ($p>0$) moment exponentially stable if there is a pair of positive constants $K$ and $c$ such that 
$  \EE |z(t)|^p \le K |z(0) |^p e^{-c t} $ for all $ t \ge 0 $,
 which leads to 
$  
 \limsup_{t \to \infty} \frac{1}{t}  \ln (\EE |z(t)|^p) \le -c  < 0   
$
for all $z(0) \in \RR^{n+q}$.
\end{definition}
\begin{definition}   \label{AS_ExpStab}
 \cite[Definition 3.1, p119]{mao2007book} System (\ref{Compact-z}) is said to be almost surely exponentially stable if 
$
   \limsup_{t \to \infty} \frac{1}{t}  \ln |z(t)|  < 0  
$  for all $z(0 ) \in \RR^{n+q}$.
\end{definition}


\section{Lyapunov stability theory for the general class of impulsive systems}   \label{sec:stabilitySiDEs}

We dedicate this section  to   establishing by  Lyapunov methods a stability theory for  our proposed general class of SiDEs. The general class of SiDEs is formulated as a canonic form of CPSs and we shall develop a foundational theory for stability of CPSs, which will be applied to the CPS of a numerical method for SDEs. In the previous section, for simplicity, the compact form (\ref{Compact-z}) of system (\ref{SiDE-xy}) is employed to study the existence and uniqueness of solutions to the SiDEs.  Now we exploit the structure of SiDE (\ref{Compact-z}) which is composed of subsystems  (\ref{SDE_x},\ref{Impulse_x}) and (\ref{SDE_y},\ref{Impulse_y}) and show the stability of  the subsystems as well as that of  the whole system (\ref{Compact-z}).

\begin{theorem}   \label{Theorem_StableContinuous}
Suppose that  Assumption   \ref{globalLipschitz} holds. Let $V \in C^{2, 1} (\RR^n \times \RR_+; \RR_+)$ and $\tilde{V} \in C^{2, 1}  (\RR^q \times \RR_+; \RR_+)$ be a pair of candidate Lyapunov functions for the subsystems (\ref{SDE_x},\ref{Impulse_x}) 
and  (\ref{SDE_y},\ref{Impulse_y}), respectively, which satisfy
\begin{subequations}      \label{LyapunovFunctions}
\begin{align}
 & c_1 |x|^p \le V(x, t) \le c_2 |x|^p,  \label{LyapunovFunctions_x} \\
&  \tilde{c}_1 |y|^p \le \tilde{V}(y, t) \le \tilde{c}_2 |y|^p  \label{LyapunovFunctions_y}
\end{align}
\end{subequations}
 for all $(x, y, t) \in \RR^n \times \RR^q \times \RR_+$ and some positive constants $p, c_1, c_2, \tilde{c}_1, \tilde{c}_2$. Assume that there are positive constants $\al$, $\tilde{\al}_1$, $\tilde{\al}_2$, $\bt$, $\tilde{\bt}_1$, $\tilde{\bt}_2$ such that 
\begin{subequations}        \label{LV-EV-xy}
\begin{align}
 &   \ll V(x, t) \le - \al V(x, t),  
  \label{LV_1}  \\ 
 &   \tll \tilde{V} (x, y, t)  \le \tilde{\al}_1 V(x, t)  + \tilde{\al}_2 \tilde{V} (y, t) , \;\; t \in [t_k, t_{k+1})   
   \label{tLtV_1} \\
 &  \EE \big[ V (x + \Del (x, k+1), t)  \big| x \big] \le \bt V (x, t),      \label{EV_tk-1}\\
 &    \EE \big[ \tilde{V} ( y + \tDel (x, y, k+1), t ) \big| x, y \big] \le \tilde{\bt}_1  V (x, t )  + \tilde{\bt}_2   \tilde{V} (y, t )   \label{EtV_tk-1}
\end{align}
\end{subequations}
for all $(x, y, t) \in \RR^n \times \RR^q \times \RR_+$ and $k \in \NN$. Let the impulse time sequence $\{ t_k \}$ satisfy
\begin{equation}  \label{ImpulseInterval}
  \frac{ \ln \bt}{ \al} < \underline{\Delta t} \le \overline{\Delta t} < \frac{- \ln \tilde{\bt}_2} { \tilde{\al}_2}.
\end{equation}
  Then SiDE (\ref{Compact-z})   is $p$th moment exponentially stable.
\end{theorem}

{\begin{proof} According to  Lemma \ref{existence_n_uniqueness}, that   Assumption \ref{globalLipschitz} holds implies there exists a unique solution to SiDE (\ref{SDE_x}-\ref{Impulse_y}).  Let us fix, for simplicity only, any $z(0) =  [ x(0)^T  \; \, y(0)^T ]^T \in \RR^{n+q}  $ and  show the stability of the solution.  The proof is rather technical so we devide it into five steps, in which we will show: 1) the exponential stability of $x(t)$; 2) some propeties of $y(t)$; 3) the exponential stability of $y(t)$ when $|x(0) |=0$; 4) the exponential stability of $y(t)$ when $|x(0)| >0$; 5) the exponential stability of $z(t)$. Some ideas and techniques in this proof are derived from our results on input-to-state stability (ISS) of SDEs  \cite[Theorem 3.1 and Remark 3.1]{huang2009}, where  $x(t)$ is treated as disturbance in the subsystem $y(t)$.

{\it Step 1:}  Observe that (\ref{LyapunovFunctions_x}), (\ref{LV_1}) and (\ref{EV_tk-1}) as well as  $ \ln \beta  <  \al \uDt$ from (\ref{ImpulseInterval}) are a specific case of conditions (i), (ii) and (iii) of \cite[Theorem 3]{huang2022} in which $\lambda_1=\gamma_1 =\cdots = \gamma_{\bar{m}} =0$. This implies 
\begin{equation}    \label{EV_0_t_exp}
    \EE V(x(t), t)  \le   \, V((0), 0)  \, e^{- \bar{\al} t}   \quad \forall \, t \ge 0
\end{equation}
where $\bal \in (0, \al - \bar{a}) $ and $\bar{a} \in (0, \al )$ with $ \ln \beta < \bar{a} \uDt <  \al \uDt$.
By condition (\ref{LyapunovFunctions_x}),   subsystem $x(t)$, which is part 
 (\ref{SDE_x},\ref{Impulse_x})  of the system   (\ref{SiDE-xy}),    is $p$th moment exponentially stable (with Lypunov exponent no larger than  $-\bal$).

{\it Step 2:} Let us  consider the dynamics of subsystem $y(t)$, which is the other part   (\ref{SDE_y},\ref{Impulse_y}) of   system  (\ref{SiDE-xy}).  By Lemma \ref{existence_n_uniqueness} and the It{\^o} formula, 
one can derive that
\begin{eqnarray}
   \EE \tilde{V} ( y(t ), t )    = \EE \tilde{V} (y(\tilde{t}), \tilde{t}) + \int_{\tilde{t}}^t \EE \tll \tilde{V}(x(s), y(s), s) \d s    \label{EtV_t}
\end{eqnarray}
for all  $t_k \le \tilde{t} \le t < t_{k+1}$ and $k \in \NN$ while condition  (\ref{tLtV_1}) produces
\begin{equation}
 \EE \tll   \tilde{V}(y(t), t) \le    \tilde{\al}_1 \EE V(x(t), t) + \tilde{\al}_2 \EE \tilde{V} (y(t), t) 
 \label{dEtV-dt}
\end{equation}
on $  [t_k, t_{k+1})$ for all $k \in \NN$. This means that $\EE \tilde{V} ( y(t ), t )$ is right-continuous on $  [0, \infty)$ and could only  jump at impulse times $\{ t_{k+1} \}_{k \in \NN}$.  Notice condition (\ref{ImpulseInterval}) implies that
 $  \tilde{\bt}_2 e^{ \tilde{\al}_2 \, \overline{\Delta t} } < 1  $
and there is a pair of positives $ \delta  \in (0, 1- \tilde{\bt}_2)$ and $\bdel \in (0, \bar{\al} ]$ sufficiently small  for
\begin{equation}    \label{def-delta}
 ( \tilde{\bt}_2 + \delta) e^{ (\tilde{\al}_2  + \delta + \bdel) \, \overline{\Delta t} }  \le 1.
\end{equation}
It is easy to observe from (\ref{dEtV-dt})  that 
\begin{equation}
  \EE \tll   \tilde{V}(y(t), t)
\le (  \tilde{\al}_2   + \delta) \EE \tilde{V} (y(t), t)     \label{ELtV-tk-tk+1}
\end{equation}
for such $t \in [t_k, t_{k+1})$ and $k \in \NN$ that
$$
 \EE \tilde{V} (y(t), t) \ge \frac{\tilde{\al}_1}{ \delta } \EE V(x(t), t) .
$$
Similarly, one can observe from (\ref{EtV_tk-1}) that
\begin{eqnarray}
   \EE   \tilde{V}(y(t_{k+1}), t_{k+1})
\le (  \tilde{\bt}_2 + \delta)  \EE \tilde{V} (y(t_{k+1}^-), t_{k+1}^-)    \label{EtV-tk+1}
\end{eqnarray}
 whenever
$$
\EE \tilde{V} (y(t_{k+1}^-), t_{k+1}^-)  \ge   \frac{\tilde{\bt}_1}{\delta } \EE V(x(t_{k+1}^-), t_{k+1}^-).
$$

{\it Step 3:}  If $x(0)=0$ (equivalently, by (\ref{LyapunovFunctions_x}),  $V(x(0),0)=0$),  then (\ref{EV_0_t_exp})  gives   $\EE V(x(t), t) = 0$ for all $t \ge 0$. Using (\ref{tLtV_1}), (\ref{EtV_t}) and (\ref{dEtV-dt}),  one has
\begin{equation}    
  \EE \tilde{V} (y(t), t)   \le  \tilde{V} (y(0), 0) + \tilde{\al}_2 \int_0^t \EE  \tilde{V} (y(s), s) \d s  \quad  \forall \, t \in [0, t_1).  \label{EtV-EL-0-t1}
\end{equation}
This, by the Gronwall inequality  (\cite[Lemma 4.5.1, p129]{kloeden1992}, \cite[Theorem 8.1, p45]{mao2007book})  implies 
\begin{equation}
  \EE \tilde{V} (y(t), t)   \le   \tilde{V} (y(0), 0) \, e^{ \tilde{\al}_2 \, t}  \quad \forall \, t \in [0, t_1)   \label{EtV-0-t1}
\end{equation}
and, particularly,
 $ \EE \tilde{V} (y(t_1^-), t_1^-)   \le  \tilde{V} (y(0), 0) \, e^{ \tilde{\al}_2 \, t_1} $.
Conditions (\ref{EtV_tk-1}) and (\ref{def-delta}) produce
\begin{multline}
  \EE \tilde{V} (y(t_1), t_1)  \le \tilde{\bt}_2  \EE \tilde{V} (y(t_1^-), t_1^-)    \le \tilde{\bt}_2  \, \tilde{V} (y(0), 0)  \,  e^{ \tilde{\al}_2 \, t_1}    \\
 <  \tilde{V} (y(0), 0) \, e^{ -(\tilde{\al}_2 + \delta + \bdel )\, \overline{\Delta t} } \, e^{ \tilde{\al}_2 \, t_1}     
 \le   \tilde{V} (y(0), 0) e^{-( \delta + \bdel )\, \overline{\Delta t} }.  \label{EtV-t1}
\end{multline}
One can repeat the derivations (\ref{EtV-EL-0-t1})-(\ref{EtV-t1}) over the interval between any two consecutive impulse times $[t_k, t_{k+1}]$ and obtain
\begin{equation}
    \EE \tilde{V} (y(t), t)   \le    \tilde{V} (y(0), 0) e^{ \tilde{\al}_2 \, (t- t_k) -k ( \delta + \bdel) \, \overline{ \Delta t}}
\end{equation}
for all $t \in [t_k, t_{k+1})$ and $k \in \NN$. 
This implies
\begin{equation}   \label{EtV-0-t-exp-V0}
  \EE \tilde{V} (y(t ), t )   \le e^{ (\tilde{\al}_2 + \delta + \bdel )\, \overline{\Delta t} } \, \tilde{V} (y(0), 0) \, e^{ -( \delta + \bdel) \, t}  \quad \forall \, t \ge 0.
\end{equation}
 By condition (\ref{LyapunovFunctions_y}),  subsystem $y(t)$ is $p$th moment exponentially stable (with Lyapunov exponent no larger than $-( \delta + \bdel)$) when $x(0) =0$.

{\it Step 4:} We shall show the exponential stability of subsystem $y(t)$  when $|x(0)|>0$, namely, $V(x(0), 0) \ge c_1 |x(0)|> 0$. 
Recall that both   $\EE V(x(t), t)$ and $\EE \tilde{V} (y(t), t)$ are right-continuous on $  [0, \infty)$, which could only   jump at impulse times $\{ t_{k+1} \}_{k \in \NN}$.  Define a function $\bar{v}: \RR_+ \to \RR$ as 
\begin{eqnarray}   \label{def-barv}
  \bar{v}(t) =\frac{ ( \tilde{\al}_1 \vee \tilde{\bt}_1) } { \delta} \EE V(x(t), t) - \EE \tilde{V} (y(t), t)  \quad  \forall \, t \in [0, \infty)
\end{eqnarray}
 with initial value 
\begin{equation*} 
 \bar{v}(0) = \frac{ ( \tilde{\al}_1 \vee \tilde{\bt}_1) } { \delta} V(x(0), 0) - \tilde{V}(y(0), 0).
\end{equation*}
Due to the properties of $\EE V(x(t), t)$ and $\EE \tilde{V} (y(t), t)$,  $ \bar{v}(t)$ is right-continuous on  $  [0, \infty)$ and could only   jump at impulse times $\{ t_{k+1} \}_{k \in \NN}$. 
Given any $t \ge 0$, either $\bar{v}(t) \ge 0$ or $ \bar{v}(t) < 0$. So the interval $[0, \infty)$ is broken into a disjoint union of subsets $T_+ \cup T_-$, where
\begin{equation}   \label{def-S+S-}
  T_+ = \{ t \ge 0:  \bar{v}(t) \ge 0 \}, \;\;   T_- = \{ t \ge 0:  \bar{v}(t) < 0 \}.
\end{equation}
Obviously,  from (\ref{EV_0_t_exp}), one has 
\begin{equation}
   \EE \tilde{V} ( y(t), t )    \le \frac{ ( \tilde{\al}_1 \vee \tilde{\bt}_1) } { \delta}  \EE V( x(t), t)     \le  \frac{ ( \tilde{\al}_1 \vee \tilde{\bt}_1) } { \delta}   \, V(x(0), 0) \, e^{- \bar{\al} t} \quad  \forall \,  t \in T_+.   \label{EtV-tau0}
\end{equation}

Notice that $V(x(0), 0)>0$. One has $\bar{v}(0)  >0$ if $\tilde{V}(y(0), 0)=0$; othewise,  one can always choose  a sufficiently small $  \delta       $ such that 
$$
  0 < \delta  <   ( \tilde{\al}_1 \vee \tilde{\bt}_1)  \frac{V(x(0), 0)}{\tilde{V}(y(0), 0)}
$$
 and, hence, $\bar{v}(0)  >0$.  
Without loss of generality,  assume that  $\bar{v}(0)  >0$.   Therefore,  $\bar{v} (t) >0$ on $[0, \epsilon)$ for some    $\epsilon >0$, i.e., $[0, \epsilon) \subset T_+$. If $T_+ = [0, \infty)$ (i.e., $T_- =\emptyset$), by (\ref{EtV-tau0}),  the proof is complete.
 Otherwise, let us consider the right-continuous process  $\EE \tilde{V} ( y(t), t ) $ on the subset $T_-$.
Due to the right-continuity of $\bar{v} (t)$ on $  [0, \infty)$,   $\forall \, \bar{t} \in T_-$, there exists an ordered pair   $ \tau_1 (\bar{t})< \tau_2 (\bar{t}) $ such that  
\begin{equation}    \label{def-tau_1-tau_2}
  \bar{t}   \in \left[ \tau_1 (\bar{t} \,), \tau_2 (\bar{t} \,) \right)  \; {\rm and}  \;   \left( \tau_1 (\bar{t} \,), \tau_2(\bar{t}\,) \right) \subset T_-,
\end{equation}
 where  
 $\tau_1 (\bar{t}) =  \inf \{   \underline{\tau}  \le \bar{t}:    \bar{v}(\tau) <  0,  \forall \tau \in [ \,  \underline{\tau}, \bar{t}  \, ] \} $ and $
  \tau_2 (\bar{t} ) =      \sup \{   \bar{\tau} > \bar{t} :   \bar{v}(\tau) <  0,  \forall \tau \in [ \,  \bar{t}, \bar{\tau}  ) \} $.   
For convenience,   we also write $\tau_1 =  \tau_1 (\bar{t} ) $ and $\tau_2 =  \tau_2 (\bar{t} ) $ when there is no ambiguity. 
The right-continuous   $\EE \tilde{V} ( y(t), t ) $ on the interval $[ \tau_1 , \tau_2 )$   falls  into one of the three categories:
\begin{enumerate} 
\item[(C0)] There is no impulse time on $[ \tau_1, \tau_2)$.
\item[(C1)]  There is exactly one impulse time  on $[ \tau_1, \tau_2)$.
\item[(C2)]   There are at least two impulse times  on $[ \tau_1, \tau_2)$.
\end{enumerate}

Each of them is considered  as follows.

(C0) There is $k \in \NN$ such that $t_k <  \tau_1 < \tau_2 \le t_{k+1}$.  Since $\bar{v}(t)$ is right-continuous and could only  jump  at  impulse times $\{ t_{k+1} \}_{k \in \NN}$, that $t_k <  \tau_1 < \tau_2 \le t_{k+1}$  implies that  $\bar{v} (t) $ is continuous on $t= \tau_1$ and, hence, by (\ref{def-tau_1-tau_2}),  $\bar{v} (\tau_1) =0$. This means 
\begin{equation}
  \EE \tilde{V} ( y(\tau_1), \tau_1 )   =  \frac{ ( \tilde{\al}_1 \vee \tilde{\bt}_1) } { \delta}  \EE V( x(\tau_1), \tau_1 )  
   \le   \frac{ ( \tilde{\al}_1 \vee \tilde{\bt}_1) } { \delta}    \, V(x(0), 0) \, e^{- \bar{\al} \tau_1}.  \label{EtV-tau1-I}
\end{equation}
Using the Gronwall inequality, one can derive from   (\ref{EtV_t}),   (\ref{dEtV-dt}) and (\ref{ELtV-tk-tk+1}) that 
\begin{equation}
    \EE \tilde{V} ( y(t), t ) \le e^{ ( \tilde{\al}_2 + \delta ) \, (t - \tau_1) } \EE \tilde{V}( y(\tau_1), \tau_1 )       \label{EtV-tau1-tau2-I}
\end{equation}
for all $t \in [\tau_1, \tau_2)$. Notice that 
$   \tau_2 - \tau_1 < t_{k+1} -t_k \le  \overline{\Delta t} $.  
Substitution of  (\ref{EtV-tau1-I}) into (\ref{EtV-tau1-tau2-I}) leads to
\begin{multline}
    \EE \tilde{V} ( y(t), t )   <    e^{ (  \tilde{\al}_2 + \delta ) \, \overline{\Delta t} } \EE \tilde{V}( y(\tau_1), \tau_1 )       \le    \frac{ ( \tilde{\al}_1 \vee \tilde{\bt}_1) } { \delta}   \,  e^{  (  \tilde{\al}_2 + \delta )  \, \overline{\Delta t} } \, V(x(0), 0)  \, e^{- \bar{\al} \tau_1}        \\
      \le   \frac{ ( \tilde{\al}_1 \vee \tilde{\bt}_1) } { \delta}     \,  e^{ (  \bar{\al} + \tilde{\al}_2 + \delta ) \, \overline{\Delta t} } \,V(x(0), 0)  \, e^{- \bar{\al} t}   \quad \forall \,  t \in [\tau_1, \tau_2).     \label{EtV-tau1-t-tau2-I}
\end{multline} 

(C1) There is exactly one impulse time $t_k$ on $[ \tau_1, \tau_2)$, where  $k$ is some positive integer since $[0, \epsilon) \subset T_+$. There are two cases:  (C1a) $\tau_1 < t_k$ and  (C1b) $\tau_1 = t_k$.

(C1a) There is  $k \ge 1$ such that $ t_{k-1} < \tau_1 < t_k < \tau_2 \le t_{k+1}$. As above, this means  that $\bar{v} (t) $ is continuous on $t= \tau_1$ and  $\bar{v} (\tau_1) =0$.  So one has   (\ref{EtV-tau1-I}) at $t=\tau_1$ and  (\ref{EtV-tau1-t-tau2-I})  for all $t \in [ \tau_1, t_k)$. But, from (\ref{EtV-tk+1}) and (\ref{def-delta}),  one has
\begin{multline}
     \EE \tilde{V} ( y(t_k), t_k )   \le    (\tilde{\bt}_2 + \delta) \EE \tilde{V} ( y(t_k^-), t_k^- )     \le  (\tilde{\bt}_2 + \delta) e^{ ( \tilde{\al}_2 + \delta ) \, (t_k - \tau_1) } \EE \tilde{V}( y(\tau_1), \tau_1 )      \\
     \le   (\tilde{\bt}_2 + \delta) e^{ ( \tilde{\al}_2 + \delta ) \, \overline{\Delta t}  } \EE \tilde{V}( y(\tau_1), \tau_1 )      \le   e^{ - \bdel \, \overline{\Delta t} } \EE \tilde{V}( y(\tau_1), \tau_1 ).     \label{EtV-tk-II}
\end{multline}  
By the Gronwall inequality,  inequalities (\ref{EtV_t}),  (\ref{dEtV-dt}), (\ref{ELtV-tk-tk+1}) and (\ref{EtV-tk-II}) produce
\begin{multline}
    \EE \tilde{V} ( y(t), t )  \le  e^{( \tilde{\al}_2 + \delta ) (t - t_k)}  \EE \tilde{V} ( y(t_k), t_k )    \le     e^{(  \tilde{\al}_2  + \delta ) \, \overline{ \Delta t} }  \EE \tilde{V}( y(t_k), t_k)     \le   e^{(  \tilde{\al}_2  + \delta -\bdel ) \, \overline{ \Delta t} }   \EE \tilde{V}( y(\tau_1), \tau_1 )    \label{EtV-tk-tau2j+2-II}
\end{multline}
for all $t \in [t_k, \tau_2)$. Therefore, when $\tau_1 < t_k < \tau_2$, combination of (\ref{EtV-tau1-I}), (\ref{EtV-tau1-t-tau2-I}) and (\ref{EtV-tk-tau2j+2-II}) imply that (\ref{EtV-tau1-t-tau2-I}) holds for  all $t \in [\tau_1, \tau_2)$.

(C1b) By  the property of $\bar{v}(t)$ and definition of $\tau_1$,  that $\tau_1 = t_k$ implies $\bar{v} (t_k^-) \ge 0$ and hence
\begin{equation}    \label{EVy-0}
     \EE \tilde{V} ( y(t_k^-), t_k^- ) \le  \frac{ ( \tilde{\al}_1 \vee \tilde{\bt}_1) } { \delta}  \EE V( x(t_k^-), t_k^- ). 
\end{equation}
Combining (\ref{EtV_tk-1}), (\ref{EV_0_t_exp}) and (\ref{EVy-0}) produces
\begin{equation}  
  \EE \tilde{V} ( y(\tau_1  ), \tau_1  )  
    \le     \big( \tilde{\bt}_1 + \frac{ ( \tilde{\al}_1 \vee \tilde{\bt}_1)  } { \delta} \,\tilde{\bt}_2  \big) \, \EE V( x(t_k^-), t_k^- )      \le    \big( \tilde{\bt}_1 + \frac{ ( \tilde{\al}_1 \vee \tilde{\bt}_1)  } { \delta} \,\tilde{\bt}_2 \big)    \, V(x(0), 0)  \, e^{- \bar{\al} \tau_1 }  
     \label{EtV-tau1-II-ii}
 \end{equation}
and hence 
\begin{multline*}         
   \EE \tilde{V} ( y(t  ), t  )    \le   e^{ ( \tilde{\al}_2 + \delta ) \, (t-\tau_1) }    \EE \tilde{V} ( y(\tau_1  ), \tau_1  )   \le   \big( \tilde{\bt}_1 + \frac{ ( \tilde{\al}_1 \vee \tilde{\bt}_1) } { \delta} \,\tilde{\bt}_2 \big) \,  e^{ ( \tilde{\al}_2 + \delta ) \, (t-\tau_1) } \,\EE V( x(\tau_1^-), \tau_1^- )     \\ 
    \le   \big( \tilde{\bt}_1 + \frac{ ( \tilde{\al}_1 \vee \tilde{\bt}_1) } { \delta} \,\tilde{\bt}_2 \big) \,    e^{ ( \tilde{\al}_2 + \delta ) \,\overline{\Delta t} } \,V(x(0), 0)  \, e^{- \bar{\al} \tau_1}     \le    \big( \tilde{\bt}_1 + \frac{ ( \tilde{\al}_1 \vee \tilde{\bt}_1) } { \delta} \,\tilde{\bt}_2 \big) \,    e^{ (  \bar{\al}+\tilde{\al}_2 + \delta ) \,\overline{\Delta t} } \,V(x(0), 0)  \, e^{- \bar{\al} t}     
\end{multline*}
for all $t \in [\tau_1, \tau_2)$. 
This combined  with (\ref{EtV-tau1-t-tau2-I}) yields 
\begin{equation}
   \EE \tilde{V} ( y(t  ), t  )    \le  K  \,V(x(0), 0)  \, e^{- \bar{\al} t}  \quad \forall \,  t \in [\tau_1, \tau_2)    \label{EtV-II}
\end{equation}
 when  there is only one impulse time on the interval $ [\tau_1, \tau_2)$, where $K$ is a positive constant
\begin{equation}     \label{def-K} 
        K = \Big( \frac{ ( \tilde{\al}_1 \vee \tilde{\bt}_1) } { \delta} \, \vee \, \big( \tilde{\bt}_1 + \frac{ ( \tilde{\al}_1 \vee \tilde{\bt}_1) } { \delta} \,\tilde{\bt}_2 \big) \Big)   \,    e^{ (  \bar{\al}+\tilde{\al}_2 + \delta ) \,\overline{\Delta t} }.
\end{equation}

(C2) There are at least two impulse times  on $[ \tau_1, \tau_2)$.
For any two consecutive impulse times $t_k$ and $t_{k+1}$ on $[ \tau_1, \tau_2)$, using the  reasoning as above, one can derive that
\begin{equation}    \label{EtV-tk-t-III}
   \EE \tilde{V} ( y(t  ), t  )   \le  e^{ ( \tilde{\al}_2 + \delta ) \, (t-t_k) }   \,  \EE \tilde{V} ( y(t_k  ), t_k ) 
\end{equation}
for all $ t \in [t_k, t_{k+1})$ and then
\begin{multline}
 \EE \tilde{V} ( y(t_{k+1}  ), t_{k+1}  )     \le      ( \tilde{\bt}_2 + \delta )    \,  \EE \tilde{V} ( y( t_{k+1}^-  ), t_{k+1}^-  )     \le    ( \tilde{\bt}_2 + \delta )  \,e^{ ( \tilde{\al}_2 + \delta ) \, (t_{k+1}-t_k) }  \,  \EE \tilde{V} ( y(t_k  ), t_k )    \\
    \le    ( \tilde{\bt}_2 + \delta )  \,e^{ ( \tilde{\al}_2 + \delta ) \,  \overline{ \Delta t}  }  \,  \EE \tilde{V} ( y(t_k  ), t_k )    \le   e^{ - \bdel \,  \overline{ \Delta t} } \EE \tilde{V} ( y(t_k  ), t_k ).      \label{EtV-tk-tk+1-III}
\end{multline}

Denote by $  t_k < \cdots < t_{\bar{k}+1} <  \cdots$ the  impulse times  on $[ \tau_1, \tau_2)$, where $\bar{k} \ge k \ge 1$. 
Let us   consider $ \EE \tilde{V} ( y(t  ), t  )$ on the interval $[ t_k, \tau_2 )$.
By (\ref{EtV-tk-t-III}) and (\ref{EtV-tk-tk+1-III}), one obtains
\begin{equation}
   \EE \tilde{V} ( y(t  ), t  ) \le  e^{ ( \tilde{\al}_2 + \delta ) \, (t-t_{ \bar{k} }) - (\bar{k}-k) \, \bdel \, \overline{ \Delta t} }  \, \EE \tilde{V} ( y(t_k  ), t_k )  \quad \forall \, t \in [ t_{ \bar{k} },  \, t_{ \bar{k}+1}  \wedge \tau_2 )
\end{equation}
 and, therefore,   
\begin{equation}   \label{EtV-tk-t-tau-2j+2-III}
   \EE \tilde{V} ( y(t  ), t  ) \le  e^{ ( \tilde{\al}_2 + \delta + \bdel ) \,   \overline{ \Delta t} -   \, \bdel \, (t- t_k)  }  \, \EE \tilde{V} ( y(t_k  ), t_k )  \quad \forall \,  t \in [ t_k,     \tau_2 ).
\end{equation}

Recall that $0 < \bdel \le \bar{\al}$ and  $0 \le t_k - \tau_1 \le \overline{ \Delta t}$. Again, there are two cases: $\tau_1 < t_k$ and $\tau_1 = t_k$.
In the case where $\tau_1 < t_k$, from (\ref{EtV-tau1-I}), (\ref{EtV-tk-II}) and (\ref{EtV-tk-t-tau-2j+2-III}), one has
\begin{multline}
    \EE \tilde{V} ( y(t  ), t  )  \le e^{ ( \tilde{\al}_2 + \delta + \bdel ) \,   \overline{ \Delta t} -   \, \bdel \, (t- t_k)  } \, e^{ - \bdel \, \overline{ \Delta t} } \EE \tilde{V} (y(\tau_1) , \tau_1)   \\
   \le \frac{ ( \tilde{\al}_1 \vee \tilde{\bt}_1) } { \delta}    \, e^{ ( \tilde{\al}_2 + \delta + \bdel ) \,   \overline{ \Delta t} }\, V(x(0), 0) \, e^{- (\bar{\al} \tau_1 + \bdel \, \overline{ \Delta t} - \bdel  t_k) - \bdel \, t }     \\
   \le \frac{ ( \tilde{\al}_1 \vee \tilde{\bt}_1) } { \delta}    \, e^{ ( \tilde{\al}_2 + \delta + \bdel ) \,   \overline{ \Delta t} }\, V(x(0), 0) \, e^{- \bdel \,  (  \tau_1 +   \overline{ \Delta t} -  t_k) - \bdel \, t }       \le \frac{ ( \tilde{\al}_1 \vee \tilde{\bt}_1) } { \delta}     \, e^{ ( \tilde{\al}_2 + \delta + \bdel ) \,   \overline{ \Delta t} }\, V(x(0), 0) \, e^{-   \bdel \, t }    
\end{multline}
for all $t \in [ t_k,     \tau_2 )$  and then, by (\ref{EtV-tau1-t-tau2-I}), 
\begin{equation}    \label{EtV-t-III-i}
\EE \tilde{V} ( y(t  ), t  )  \le \frac{ ( \tilde{\al}_1 \vee \tilde{\bt}_1) } { \delta}     \, e^{ ( \tilde{\al}_2 + \delta + \bdel ) \,   \overline{ \Delta t} }\,  V(x(0), 0) \, e^{-   \bdel \, t }   \quad   \forall   \, t \in [ \tau_1,     \tau_2 ) .  
\end{equation}
In the other case where $\tau_1 = t_k$, substituting (\ref{EtV-tau1-II-ii}) into (\ref{EtV-tk-t-tau-2j+2-III}) gives
\begin{multline}     \label{EtV-t-III-ii}
     \EE \tilde{V} ( y(t  ), t  )  
 \le     \big( \tilde{\bt}_1 + \frac{ ( \tilde{\al}_1 \vee \tilde{\bt}_1)  } { \delta} \,\tilde{\bt}_2 \big)  \,
    e^{ ( \tilde{\al}_2 + \delta + \bdel ) \,   \overline{ \Delta t} }     V(x(0), 0)  \, e^{- ( \bar{\al} - \bdel) \, \tau_1  -   \, \bdel \, t  }    \\
   \le     \big( \tilde{\bt}_1 + \frac{ ( \tilde{\al}_1 \vee \tilde{\bt}_1)  } { \delta} \,\tilde{\bt}_2 \big)  \,
   e^{ ( \tilde{\al}_2 + \delta + \bdel ) \,   \overline{ \Delta t} }\, V(x(0), 0)  \, e^{-   \bdel \, t  }    \quad \forall \,   t \in [ \tau_1, \tau_2 ).
\end{multline}  
Combination of (\ref{EtV-t-III-i}) and (\ref{EtV-t-III-ii}) yields
\begin{equation}    \label{EtV-t-III}
     \EE \tilde{V} ( y(t  ), t  )  \le   K V(x(0), 0)  \, e^{- \bar{\delta} \,  t}    
\end{equation}
for all $t \in [ \tau_1, \tau_2 )$ on which  there are at least two impulse times, where  $K$ is the positive constant given by (\ref{def-K}).
%
From inequlities (\ref{EtV-tau1-t-tau2-I}), (\ref{EtV-II}) and (\ref{EtV-t-III}), one has
\begin{equation}    \label{EtV-t-tau1-tau2}
     \EE \tilde{V} ( y(t  ), t  )  \le   K V(x(0), 0)  \, e^{- \bar{\delta} \,  t}    \quad  \forall \, t \in T_-.
\end{equation}
 Combining (\ref{EtV-tau0}) and (\ref{EtV-t-tau1-tau2}),  one can  conclude that 
\begin{equation}    \label{EtV-t-infty}
     \EE \tilde{V} ( y(t  ), t  )  \le   K  V(x(0), 0)  \, e^{- \bar{\delta} \,  t}    \quad \forall \, t \ge 0.
\end{equation}
 By condition (\ref{LyapunovFunctions_y}), this means that subsystem $y(t)$,  the other part of the system,  is also $p$th moment exponentially stable (with Lyapunov exponent no larger than $-\bar{\delta}$) when  $|x(0)| >0$.

{\it Step 5:} We have shown the $p$th moment exponential stability of  $x(t)$ by  (\ref{EV_0_t_exp}) and  that of $y(t)$ by (\ref{EtV-0-t-exp-V0}) and (\ref{EtV-t-infty}) when $|x(0) | =0$ and $|x(0) | >0$, respectively. 

Note that   $ z (t) = [ x^T(t) \; \, y^T(t) ]^T = [ x^T(t) \;  \, 0^T ]^T + [ 0^T  \; \,  y^T(t) ]^T $
and, hence,
 $   | z(t) | \le |x(t) | + |y(t)|   $
for all $t \ge 0$. It is easy to see that
\begin{equation}   \label{z_p}
    | z(t) |^p \le ( |x(t) | + |y(t)| )^p  \le  k_p  ( |x(t) |^p + |y(t)|^p ),
\end{equation}
where $k_p=1$ when $0<p<1$ and $k_p = 2^{p-1}$ when $ p \ge 1$. Obviously, in the case where $|x(0) | =0$ and then $\EE |x(t)|^p =0$ for all $ t \ge 0$,  inequalities  (\ref{EtV-0-t-exp-V0}) and (\ref{z_p}) as well as $| z(0) | = | y(0) | $ give  
\begin{equation}
\EE |z(t)|^p \le k_p \EE |y(t) |^p \le \frac{k_p} { \tilde{c}_1}  e^{ (\tilde{\al}_2 + \delta + \bdel )\, \overline{\Delta t} } \, \tilde{V} (y(0), 0) \, e^{ -( \delta + \bdel) \, t}   
\le   \frac{k_p \tilde{c}_2} { \tilde{c}_1}  e^{ (\tilde{\al}_2 + \delta + \bdel )\, \overline{\Delta t} } |z(0)|^p e^{ -( \delta + \bdel) \, t}  
   \label{Ez-p-exp_x0} 
\end{equation} 
for all $ t \ge 0 $.    In the general case where  $|x(0) | >0$,  inequalities  (\ref{LyapunovFunctions}), (\ref{EV_0_t_exp}), (\ref{EtV-t-infty}) and (\ref{z_p}) produce
\begin{equation}    \label{Ez-p-exp}
    \EE |z (t)|^p  \le  k_p \, ( 1 \vee \bt) \, \frac{c_2}{c_1}  |x(0)|^p  \, e^{- \bar{\al} t}    + k_p \,  K \, \frac{c_2}{c_1}  |x(0)|^p  \, e^{- \bar{\delta} \,  t}    \le  \bar{K}_p |z(0)|^p e^{- \bar{\delta} \,  t}   
\end{equation}
for all $t \ge 0$, where $K$ is the positive constant given  by (\ref{def-K}) and $  \bar{K}_p$ is a positive constant  
$$
    \bar{K}_p = \frac{k_p \, c_2} {c_1} (  ( 1 \vee \bt)  + K).
$$
So  (\ref{Ez-p-exp_x0}) and (\ref{Ez-p-exp}) mean that system  (\ref{Compact-z}) (i.e., system (\ref{SiDE-xy})) is $p$th moment exponentially stable (with Lyapunov exponent no larger than $- \bdel$).  
\end{proof}

\begin{remark}
  Notice that, in Theorem \ref{Theorem_StableContinuous}, the continuous dynamics of subsystem $x(t)$ stabilizes the subsystem,  though the discrete one could destabilize it, while the discrete dynamics of subsystem $y(t)$ stabilizes the subsystem,  though the continuous one could destabilize it,  which results in the exponential stability of the both subsystems and hence that of the whole system $z(t) = [ x^T(t) \; y^T(t)]^T$.  Similarly, one can obtain  a stability criterion for the case where the impulses  stabilize the physical subsystem $x(t)$ as the continuous dynamics could destabilize it (see \cite[Theorem 2]{huang2022}) while the conditions on the subsystem $y(t)$ are kept the same as those in Theorem \ref{Theorem_StableContinuous}.
\end{remark}

Furthermore,  under  Assumption \ref{globalLipschitz},  we have the following result for the almost sure stability.
\begin{theorem}   \label{Theorem_AlmostSure}
If Assumption \ref{globalLipschitz} holds, 
then the  $p$th ($p>0$) moment exponential stability of SiDE  (\ref{Compact-z}) (i.e., system (\ref{SiDE-xy})) implies that  it is also almost surely exponentially stable. 
\end{theorem}

The proof is similar to that of \cite[Theorem 4.2, p128]{mao2007book} and, therefore,  is omitted.

\section{The cyber-physical systems of numerical methods }    \label{sec:numericalCPS}

In this section, we address the  problem (I)  of fundamental importance. We  compose a hybrid model in the form of our proposed   SiDE (\ref{SiDE-xy})  to represent the tight integration of the physical system (the SDE)  and its cyber counterpart (the numerical method). This systematic representation  can serve as a canonic form of CPS models. 

Let us consider a physical system described by SDE
\begin{equation}    \label{SDE_app} 
 \mathrm{d}  x(t) = f(x(t))\mathrm{d}t + g(x(t) )\mathrm{d} B(t)   \quad \forall \,  t \ge 0 
\end{equation}
 with initial value $x(0)  \in \RR^n$, where $f: \RR^n \to \RR^n $ and $g: \RR^n \to \RR^{ n \times m}$  satisfy the global Lipschitz condition
\begin{equation} \label{globalLipschitz_app}
|f(x ) - f(\bar{x} ) | \vee |g(x )- g(\bar{x} )|    
 \le L  |x- \bar{x}| \quad  \forall \, (x, \bar{x}) \in \RR^n \times \RR^n
\end{equation}
as well as  $f(0)=0$ and $g(0)=0$ for  study of the stability problem. 
Given a fixed parameter $\theta \in [0, 1]$, the following  numerical scheme for   SDE (\ref{SDE_app})  is called  the   stochastic theta method  \cite{higham2000,higham2002, mao2015a, mao2015b}   
\begin{equation}     \label{STMethod}
 X_{k+1} = X_k + (1 -\theta) f(X_k) \Delta t + \theta f(X_{k+1}) \Delta t     + g( X_k) \sqrt{\Delta t} \, \xi (k+1)  \quad \forall  \, k \in \NN
\end{equation}
with  initial value $X_0 =x(0)$, where $\Delta t>0$ is the constant stepsize and $\sqrt{\Delta t} \, \xi (k+1)$ is the implementation of $\Delta B_k = B( (k+1) \Delta t) -B(k \Delta t)$.   The stochastic theta method for  SDEs  is one of the popular algorithms \cite{kloeden1992,higham2001}  used in the so-called actor models \cite{lee2017},  techniques driven by software modelling and simulation tools, 
to describe and compute the  physical dynamics  (\ref{SDE_app}).  Numerical method  (\ref{STMethod})  is in essence  a cyber model of the  physical system (\ref{SDE_app}), which is a translation  of    the  physical system (\ref{SDE_app}) into discretization, the symbolism in computers,  and represents the physical dynamics in  the cyber world.

When $\theta =0$, numerical scheme (\ref{STMethod}) gives the widely-used Euler-Maruyama  method. The Euler-Maruyama method applied to SDE (\ref{SDE_app}) computes approximations $X_k \approx x(t_k)$ for all $k \in \NN$, where $t_k = k \Delta t$, by setting $X_0 =x (0)$ and forming
\begin{equation}     \label{EMMethod}
 X_{k+1} = X_k +  f(X_k) \Delta t   + g( X_k) \sqrt{\Delta t} \, \xi (k+1)  \quad \forall  \, k \in \NN.
\end{equation}
Stochastic difference equations (\ref{EMMethod}),  also called discrete-time stochastic systems \cite{huang2015},  have been intensively studied over the  past a few decades of in the age of  computers. 
In practice, it is natural to form and use some continuous-time extension of the discrete approximation $\{ X_k \}$  such as \cite{higham2002,mao2015b}
\begin{equation}     \label{def-X-app}
     X(t) = \sum_{k=0}^\infty X_k \ONE_{[t_k, t_{k+1})} (t)  \quad   \forall   \, t \ge 0
\end{equation}
where $ \ONE_T$ is the indicator function of set $T$. This is a simple step process of the equidistant Euler-Maruyama approximations so its sample paths are continuous on $( t_k, t_{k+1})$ for each $k \in \NN$ and right-continuous on $[0, \infty)$. 

This paper consider the widely-used Euler-Maruyama  method (\ref{EMMethod})-(\ref{def-X-app})   the cyber system,  which is virtually a representative of    the  physical system  (\ref{SDE_app}) in the cyber world. Other numerical schemes, or say, other translations can also be employed to represent the physical system in the cyber world in future work. 
This section is to unveil the inherent relationship between a physical system and its cyber couterpart. 
Consider the  process $y(t)$ of difference between  the exact solution $x(t)$ of the physical system (\ref{SDE_app})  and the numerical solution $X(t)$ by its cyber counterpart   (\ref{EMMethod})-(\ref{def-X-app})
\begin{equation}     \label{def-y-app}
    y(t) = x(t) -  X(t)  \quad  \forall \,  t \ge 0
\end{equation} 
with initial value $y(0)=x(0)-X(0) =0$. Notice that $x(t)$ is a  process of continuous paths and $X(t)$  a simple step process. This implies  that $y(t)$ is right-continuous  on $[ 0, \infty) $ and could only   jump at   $\{ t_{k+1} \}_{k \in \NN}$. According to the  scheme (\ref{EMMethod})-(\ref{def-X-app}),  the jump of $y(t)$ at $t=t_{k+1}$ is
\begin{multline}
   y(t_{k+1}) - y(t_{k+1}^-)   =  x(t_{k+1}) - X(t_{k+1}) - \left(  x(t_{k+1}^-) - X(t_{k+1}^-) \right)    \\
 =  X(t_{k+1}^-) - X(t_{k+1})    = -   f(X_k) \Delta t  - g( X_k) \sqrt{\Delta t} \, \xi (k+1)      \\
    =  - f( X(t_{k+1}^-) ) \Delta t - g( X(t_{k+1}^-) ) \sqrt{\Delta t} \, \xi (k+1)     \\
    =  - f( x(t_{k+1}^-) - y(t_{k+1}^-)) \Delta t    - g( x(t_{k+1}^-) -y(t_{k+1}^-) ) \sqrt{\Delta t} \, \xi (k+1)   \quad  \forall \, k \in \NN
\end{multline}
 since
  $ X(t) = x(t) - y(t) = X_k $
for all $t  \in [ t_k, t_{k+1})$ and $k \in \NN$. 

The integrated dynamics of the physical system (\ref{SDE_app}) and the process (\ref{def-y-app}) of difference  is described by the following hybrid system in the form of SiDEs
\begin{subequations}  \label{SiDE-xy_app}
\begin{align}
& \mathrm{d}  x(t) = f (x(t) )\mathrm{d}t + g(x(t)  )\mathrm{d} B(t)   \label{SDE_x_app} \\
& \d y(t)  = f(x(t) )\mathrm{d}t + g(x(t)  )\mathrm{d} B(t)  
\quad t \in [t_k, t_{k+1})     \label{SDE_y_app}   \\ 
&\tilde{\Delta} (x(t_{k+1}^-), y(t_{k+1}^-), k+1) := y(t_{k+1}) - y(t_{k+1}^-)  \nonumber  \\
 & \qquad   {}       = - f (x(t_{k+1}^-)-y(t_{k+1}^-)) \Delta t     - g (x(t_{k+1}^-)- y(t_{k+1}^-)) \sqrt{\Delta t} \, \xi (k+1)  \quad \, k \in \NN
\label{Impulse_y_app} 
\end{align}
\end{subequations}
 with initial data $x(0)  \in \RR^n$ and $ y(0)=   x(0)-X(0)  =0 $. 
Clearly, this is a special  case of  our proposed  impulsive system  (\ref{SiDE-xy}) in which $h_f (\cdot, \cdot) \equiv 0$ and $h_g (\cdot, \cdot) \equiv 0$, that is,  no impulse is imposed on the physical subsystem (\ref{SDE_x_app}). 
We construct the CPS  (\ref{SiDE-xy_app}) of the widely-used Euler-Maruyama  method (\ref{EMMethod})-(\ref{def-X-app}) for the SDE  (\ref{SDE_app}), which  is a seamless, fully synergistic integration of the physical system (\ref{SDE_app})  and its cyber counterpart (\ref{EMMethod})-(\ref{def-X-app}).  The CPS   not only provides a holistic view of the physical system and its cyber counterpart  but also reveals their intrinsic relationship that they are not two separate systems but the components of an integrated system.

Recall that the SDE describes our knowledge of  the physical dynamics while the numerical method is the cyber representive, namely, the translation of our  knowledge  in the cyber world.  As a result, the CPS  (\ref{SiDE-xy_app})  is the combination of  our knowledge of the physical system and the cyber  representative  as well as the  simulation sequence $\{ \xi (k) \}_{k \in \NN} $.  Moreover, the CPS   clearly shows that the numerical solution is driven by the SDE and the numerical method as well as the  simulation sequence while the exact solution  is, of course, conducted by the SDE itself only.  
Usually, to control the underlying physical system,   our knowledge and the cyber representive of the system are  both utilized in the synthesis of the CPS. This leads to the resulting CPS with $y(t)$ involved in the dynamics equation of $x(t)$ as well, which is considered in the study of stabilization  problems \cite{huang_partII}. 

Notice that the   physical system (\ref{SDE_app}) is subject to no impulse and  its  cyber-physical model (\ref{SiDE-xy_app}) is a specific case of our proposed canonic form   (\ref{SiDE-xy}) of CPSs in which   $q =n$, $f(x, t) = f(x)$, $g(x, t) = g(x)$, $\tilde{f} (x, y, t)= f (x)$,  $\tilde{g} (x, y, t)= g (x)$, $h_f (x, k) =0$, $h_g (x, k) =0$, $\tilde{h}_f (x, y, k) =-f (x-y) \Delta t $, $\tilde{h}_g (x, y, k) = -g (x- y) \sqrt{\Delta t} $ and $t_k =k \Delta t$. 
Consequently, the infinitesimal generators (\ref{LV}) and (\ref{tLtV})  associated with (\ref{SDE_x_app}) and (\ref{SDE_y_app}) are of the specific forms 
\begin{eqnarray}    
 \ll V (x) =   V_x (x) f(x ) + \frac{1}{2} {\rm trace} \left[ g^T (x) V_{xx} (x) g(x) \right] ,   \label{LV_1-x}  \\
   \ll  \tilde{V} (x, y) 
   =  \tilde{V}_y (y) f(x) + \frac{1}{2} {\rm trace} \left[ g^T (x) \tilde{V}_{yy} (y) g(x) \right],  \label{LV_1-y}
\end{eqnarray}
respectively. It is easy to see that  Assumption 1    holds since both $f$ and $g$ satisfy the global Lipschitz condition (\ref{globalLipschitz_app}). Accordint to Lemma \ref{existence_n_uniqueness}, there exists a unique (right-continuous) solution  to SiDE (\ref{SiDE-xy_app}) on $t \ge 0$ and the solution belongs to  $\mm^2 ([0, T]; \RR^{2n})$ for all $T \ge 0$.  Moreover,   the results of our established stability theory  for the general class (\ref{SiDE-xy}) of  SiDEs,  say, Theorem \ref{Theorem_StableContinuous} and Theorem \ref{Theorem_AlmostSure}  apply to the CPS   (\ref{SiDE-xy_app}).

\section{Exponential stability of the cyber-physical systems}   \label{sec:numericalSDEs}

 The  CPS  (\ref{SiDE-xy_app}) of the Euler-Maruyama method (\ref{EMMethod})-(\ref{def-X-app}) for the SDE   (\ref{SDE_app})  consists of the physical  subsystem and the cyber subsystem. The key questions   (Q1) and (Q2) naturally arise.
In this section, we address the problem (II) of fundamental importance and prove positive results to the  key questions (Q1) and (Q2). These fundamental results and their applicaton to linear systems comprise a foundational theory  of the CPSs of numerical methods for SDEs.  

Let us  begin with   the test problem (Q1)  of  the CPS   (\ref{SiDE-xy_app}), to which    Theorem \ref{Theorem_StableContinuous} and Theorem \ref{Theorem_AlmostSure} can be directly applied.
Under some conditions (see \cite{gikhman1996, khas2012}),  a seminal converse Lyapunov theorem \cite[Theorem 5.12, p172]{khas2012} states that, if the SDE  (\ref{SDE_app})  is $p$th moment exponentially stable, there is a Lypunov function that proves the exponential stability of the dynamical system. One may postulate that  the Lyapunov function for the physical subsystem  (\ref{SDE_x_app}) could be used to construct  a candidate Lyapunov function for  the subsystem (\ref{SDE_y_app},\ref{Impulse_y_app}) 
due to their interrelation. 
The direct application of Theorem \ref{Theorem_StableContinuous}   to the CPS (\ref{SiDE-xy_app})  shows that the CPS  (\ref{SiDE-xy_app}) and hence the cyber system  (\ref{EMMethod})-(\ref{def-X-app})  share  the  exponential stability with the physical  system (\ref{SDE_app}).  

\begin{theorem}   \label{Theorem_StableContinuous_app}
 Let $V \in C^2(\RR^n ; \RR_+)$   
be a  candidate Lyapunov function for both subsystems  (\ref{SDE_x_app}) and  (\ref{SDE_y_app},\ref{Impulse_y_app})  and, moreover,
\begin{equation}      \label{LyapunovFunctions_app}
  c_1 |x|^p \le V(x ) \le c_2 |x|^p   \quad \forall \,  x \in \RR^n
\end{equation}
  for   some positives $p, c_1, c_2$. Assume that there are positives $\al$, $\tilde{\al}_1$, $\tilde{\al}_2$,   $\tilde{\bt}_1$, $\tilde{\bt}_2$ such that 
\begin{subequations}
\begin{align}
 &   \ll V(x ) \le - \al V(x)      \quad    \forall \,  x \in \RR^n     \label{LV_app-x}   \\     
 & \tll V (x, y)  \le \tilde{\al}_1 V(x)  + \tilde{\al}_2 V (y)   \quad    \forall \,  t \in [t_k, t_{k+1})    \label{LV_app-y} \\ 
 &    \EE V ( y + \tDel (x , y, k+1)) \le \tilde{\bt}_1   V (x)  + \tilde{\bt}_2  V (y)    \quad \forall \, (x, y) \in \RR^n \times \RR^n  \label{EV_tk-app}
\end{align}
\end{subequations}
for all   $k \in \NN$. If the stepsize 
\begin{equation}  \label{ImpulseInterval_app}
              \Delta t < \frac{- \ln \tilde{\bt}_2} { \tilde{\al}_2} ,
\end{equation}
then the CPS (\ref{SiDE-xy_app}) is $p$th moment exponentially stable and is also    almost surely exponentially stable. Moreover, the cyber system  (\ref{EMMethod})-(\ref{def-X-app}) 
 shares the $p$th moment exponential stability with its underlying physical  system  (\ref{SDE_app}) and, hence, it is also almost surely exponentially stable.
\end{theorem}

\begin{proof}
 From Theorem  \ref{Theorem_StableContinuous} and Theorem \ref{Theorem_AlmostSure}, it   follows that CPS  (\ref{SiDE-xy_app}), which is a specific case of system (\ref{SiDE-xy}) (namely,    (\ref{Compact-z})),  is $p$th moment exponentially stable and is also almost surely exponentially stable.

Notice that the state of cyber system  (\ref{EMMethod})-(\ref{def-X-app}) is the difference  $X(t) = x(t) - y(t)  $ of the subsystems  (\ref{SDE_x_app}) and  (\ref{SDE_y_app},\ref{Impulse_y_app}).   Therefore,
\begin{equation}    \label{Xp-zp}
 |X(t)|^p \le k_p ( |x(t) |^p + |y(t)|^p)  \le 2 k_p |z(t) |^p \quad  \forall \, t \ge 0
\end{equation}
where $k_p=1$ if $0<p<1$, and $k_p = 2^{p-1}$ if $ p \ge 1$. 
  This immediately implies that  the cyber system  (\ref{EMMethod})-(\ref{def-X-app}) is $p$th moment exponentially stable and  is also almost surely exponentially stable.
\end{proof}

 This means that, if the underlying physical system  (\ref{SDE_app}) is $p$th moment exponentially stable, the CPS  (\ref{SiDE-xy_app})  and, hence, the numerical method  (\ref{EMMethod})-(\ref{def-X-app})  reproduce the $p$th moment exponential stability of the physical dynamics  when the conditions in Theorem \ref{Theorem_StableContinuous_app} hold. 
The ability of the cyber system (the numerical method) to reproduce the mean-square exponential  stability of its underlying physical system (the SDE) has been studied in \cite{higham2000,higham2003}, 
Here, we investigate on the ability of the cyber system  (\ref{EMMethod}-\ref{def-X-app})  to reproduce the mean-square exponential  stability of the physical system   (\ref{SDE_app}) in our proposed framework of CPS  (\ref{SiDE-xy_app}).
A result on  mean-square exponential stability is then derived from Theorem \ref{Theorem_StableContinuous_app}  as follows, in which the Lyapunov function for mean-square exponential stability of the underlying physical system  (\ref{SDE_app}) also proves  the mean-square exponential stability  of its cyber counterpart   (\ref{EMMethod})-(\ref{def-X-app}) as well as  that of the CPS  (\ref{SiDE-xy_app}).

\begin{theorem}   \label{Theorem-MeanSquare}
Let the candidate Lyapunov function $V \in C^2(\RR^n ; \RR_+)$ for physical system  (\ref{SDE_app}) be a quadratic function
\begin{equation}   \label{LyapunovFunction-MeanSquare}
   V (x) = x^T P x
\end{equation}
for some positive definite matrix $P  \in \RR^{n \times n}$.
Assume there are  positives  $\bal $ and $\overline{\Delta t} $ with $\bal \, \overline{\Delta t} <1$  such that 
\begin{equation}
     \ll V(x ) + \overline{\Delta t} \, V(f(x))  \le - \bal V(x)   \label{LV-MeanSquare}   
\end{equation}
for all $x \in \RR^n$. Then the  CPS  (\ref{SiDE-xy_app}) with $\Delta t \in (0, \overline{\Delta t} \,]$ is mean-square exponentially stable and   is also almost surely exponentially stable. Moreover, the cyber  system  (\ref{EMMethod})-(\ref{def-X-app}) with  $\Delta t \in (0, \overline{\Delta t} \,]$  shares the mean-square exponential stability with its underlying physical system  (\ref{SDE_app}) and, hence, it is also almost surely exponentially stable.
\end{theorem}

\begin{proof}  It will follow the conclusion from Theorem \ref{Theorem_StableContinuous_app} if one  shows that  conditions  (\ref{LyapunovFunctions_app})-(\ref{ImpulseInterval_app}) of Theorem \ref{Theorem_StableContinuous_app} are satisfied with $p=2$  for the
 CPS  (\ref{SiDE-xy_app}).
Since the quadratic function   (\ref{LyapunovFunction-MeanSquare}) gives 
$ \lambda_m (P) |x|^2 \le V(x) \le \lambda_M (P) |x|^2$, 
 condition (\ref{LyapunovFunctions_app}) holds for positive constants $p=2$, $c_1 =\lambda_m (P)$, $c_2 = \lambda_M (P) $. 
Moreover,    (\ref{LV_app-x}) and  
(\ref{LV-MeanSquare}) are equivalent.  Obviously,   (\ref{LV-MeanSquare}) implies that  (\ref{LV_app-x}) holds for some positive constant $\alpha \ge \bal$. But    (\ref{LV_app-x}) implies that there is a sufficiently small positive number $ \overline{\Delta t}$ such that
 (\ref{LV-MeanSquare}) holds for some positive constant $\bal <  \alpha  \wedge ( 1 / \, \overline{\Delta t}  )$. According to \cite[Theorem 4.4, p130]{mao2007book} and  \cite[Theorem 4.2, p128]{mao2007book}, system (\ref{SDE_app}) is mean-square exponentially stable and    also almost surely exponentially stable.

By the It{\^o} lemma, \cite[Lemmas 1 and 2]{huang2010} and global Lipschitz condition (\ref{globalLipschitz_app}),
\begin{multline}    \label{tLV-MeanSquare}
   \tll V (x, y)= 2 y^T P f(x) +  trace  \big[  g^T (x) P g(x)  \big] \\ 
 \le \tilde{\al}_2 y^T P y + \tilde{\al}_2^{-1} f^T (x) P f( x)   +  \lambda_M (P) \, trace \big[  g^T(x)  g(x)  \big]     \\
  \le    \tilde{\al}_2^{-1}  \lambda_M(  P) | f(x)|^2 + \lambda_M (P) \, L^2 |x|^2  + \tilde{\al}_2 \tilde{V} (y)    \\
  \le    (  \tilde{\al}_2^{-1} +1)  \lambda_M(  P)  L^2 \, |x|^2 + \tilde{\al}_2 \tilde{V} (y)   
  \le \tilde{\al}_1 x^T P x  + \tilde{\al}_2 \tilde{V} (y)  = \tilde{\al}_1 V( x)  + \tilde{\al}_2 \tilde{V} (y),
\end{multline}
where  
$
\tilde{\al}_1 = \frac{  (1+ \tilde{\al}_2) \,  \lambda_M(  P)  L^2 }{ \tilde{\al}_2 \, \lambda_m (P) }  
$
and $\tilde{\al}_2$  given as (\ref{def-tal2}) 
below are both positive numbers. So condition  (\ref{LV_app-y})   
of Theorem  \ref{Theorem_StableContinuous_app} is satisfied. 
Note that, for any $\Delta t \in (0, \overline{\Delta t} \,]$, inequality (\ref{LV-MeanSquare}) implies
\begin{eqnarray}
     \ll V(x ) +  \Delta t \, V(f(x))  \le - \bal V(x),   \label{LV-MeanSquare-1}   
\end{eqnarray}
and (\ref{Impulse_y_app})  gives
\begin{multline*}  
   y + \tilde{\Delta} (x, y, k+1)    = y  -f  ( x  - y ) \,  {\Delta t}  - g ( x - y) \, \sqrt{\Delta t} \; \xi  (k+1)   \\
 = x -(x-y)  -f  ( x  - y ) \,  {\Delta t}  - g ( x - y) \,  
\sqrt{\Delta t} \; \xi  (k+1) .
\end{multline*}
Using inequality (\ref{LV-MeanSquare-1})  and \cite[Lemma 1]{huang2010}, 
one obtains 
\begin{multline} \label{EV-k+1-MeanSquare}
   \EE V( y+ \tilde{\Delta} (x, y, k+1))    
  = x^T P x -2 x^T P (x-y) + (x-y)^T P (x-y)   - 2  \Delta t x^T P f(x-y)    \\
  + \Delta t \Big\{ 2 (x-y)^T P f(x-y)   + trace [g^T(x-y) P g(x-y)]    + \Delta t \, f^T(x-y) P f(x-y) \Big\}     \\
  \le (1+ c^{-1}) x^T P x + (1+ c) (x-y)^T P (x-y)    - 2 \Delta t x^T P f(x-y)    + \Delta t \big[  \ll V(x-y ) +  \Delta t \, V(f(x-y)) \big]     \\
  \le  (1+ c^{-1}) V(x) + (1+c) (x-y)^T P (x-y)    - 2 \Delta t x^T P f(x-y)    - \bal \Delta t \, (x-y)^T P (x-y)       \\
   \le   (1+ c^{-1}) V(x) +  (1  +c - \bal \Delta t  )\, (x-y)^T P (x-y)     - 2 \Delta t x^T P f(x-y)      \quad \forall \, k \in \NN
\end{multline}
for all $(x, y) \in \RR^{n} \times \RR^n$, where $c $ is a positive constant  with $ c < \bal \Delta t   $ and, hence, $1+c -\bal \Delta t  <1$.  

By \cite[Lemmas 1-2]{huang2010} and global Lipschitz condition (\ref{globalLipschitz_app}), one also has
\begin{eqnarray}  
    (x-y)^T P (x-y)  & \le &  x^T P x - 2 x^T P y + y^T P y    
                             \le    (1 + b^{-1} ) x^T P x + ( 1 + b)  y^T P y,     \label{x-yPx-y} \\                               
     - 2   x^T P f(x-y)
                          &  \le  &     b^{-1} x^T P x +b f(x- y)^T P f(x-y)      
   \le     b^{-1}  x^T P x + b \lambda_M (P)  L^2 \, (x-y)^T (x-y)      \nonumber \\
 &  \le &    \big( b^{-1} + \frac{(1+b) \lambda_M (P)  L^2}{ \lambda_m (P)} \big) x^T P x    
    \; {}+ \frac{b (1+b)  \lambda_M (P)  L^2}{ \lambda_m (P)} y^T P y ,    \label{xPfx-y}
\end{eqnarray}
where $b$ is a positive constant sufficiently small for 
\begin{equation}   \label{def-tbeta2}
  \tilde{\bt}_2 :=   (1 +c - \bal \Delta t  ) (1+b) + \Delta t  \;  \frac{b (1+b)  \lambda_M (P)  L^2}{ \lambda_m (P)}  <1.
\end{equation}
Substitution of (\ref{x-yPx-y}) and (\ref{xPfx-y}) into (\ref{EV-k+1-MeanSquare}) yields
\begin{equation}
    \EE V( y+ \tilde{\Delta} (x, y, k+1))  \le \tilde{\bt}_1 V (x) + \tilde{\bt}_2 V(y) ,
\end{equation}
where 
\begin{equation*}
  \tilde{\bt}_1 =  (1+c^{-1}) + (1 +c - \bal \Delta t  ) (1+b^{-1} )  +  \Delta t   \left( b^{-1} + \frac{(1+b) \lambda_M (P)  L^2}{ \lambda_m (P)} \right)
\end{equation*}
 and $\tilde{\bt}_2$ given as (\ref{def-tbeta2}) above are both positive constants.  This is the condition (\ref{EV_tk-app}) of Theorem  \ref{Theorem_StableContinuous_app}.  

 Let  $\tilde{\al}_2$  be a positive number such that
\begin{equation}   \label{def-tal2}
    \tilde{\al}_2 < \frac{- \ln \tilde{\bt}_2} { \overline{ \Delta t}} .
\end{equation}
For instance, let  $\tilde{\al}_2 = - \ln \tilde{\bt}_2 /( 2\, \overline{ \Delta t})$ and, hence,
   $   \Delta t \le  \overline{ \Delta t}  = \frac{- \ln \tilde{\bt}_2} {2 \,\tilde{\al}_2}<  \frac{- \ln \tilde{\bt}_2} { \tilde{\al}_2} $.
 So the condition (\ref{ImpulseInterval_app}) of Theorem \ref{Theorem_StableContinuous_app} is also satisfied. 
According to Theorem \ref{Theorem_StableContinuous_app}, it follows the assertions.
\end{proof}

   In the literature \cite{higham2003,mao2015a},   to ensure that the cyber system shares the exponential stability with its underlying physical system,  the  stepsize $\Delta t$ is explicitly and severly limited by both the growth  and the rate constants of  the physical system. Although the both are related, it is only the rate constant that plays a key role in the definitions of exponential stability. It is reasonable and possible  to lessen the dependence of the stepsize $\Delta t$ on the growth constant, which itself could be very conservative  due to condition (\ref{LyapunovFunctions_app}). In Theorem \ref{Theorem-MeanSquare}, we manage to remove the explicit dependence of the stepsize $\Delta t$ on the growth constant $ \lambda_M (P) /  \lambda_m (P) $. Instead, we show that the growth constant $ \lambda_M (P) /  \lambda_m (P) $  has an   influence   on the    stepsize  $\Delta t$   through the rate-like  constant $\tilde{\bt}_2$ given by   (\ref{def-tbeta2}). This could reduce much  the restriction imposed by the growth constant.
As will be shown in Section \ref{sec:linearSDEs}, it  improves the upper bound $ \overline{ \Delta t}$ of stepsizes and eases its computation for linear systems.

Recall that Theorem \ref{Theorem_StableContinuous_app} is the direct application of Theorem \ref{Theorem_StableContinuous}   to the CPS (\ref{SiDE-xy_app}) of the numerical method (\ref{EMMethod})-(\ref{def-X-app}) for the SDE  (\ref{SDE_app}), from which Theorem \ref{Theorem-MeanSquare} is derived for the mean-square exponential stability.  
 Let us proceed to apply Theorem \ref{Theorem-MeanSquare} and study the converse question (Q2) whether one can infer that the  CPS (\ref{SiDE-xy_app}) and, hence, the physical system  (\ref{SDE_app})   are mean-square exponentially stable if the cyber system  (\ref{EMMethod})-(\ref{def-X-app}) is mean-square exponentially stable for  small stepsizes $\Delta t >0$.  
Similarly, the converse Lyapunov theorem   \cite{khas2012,vaidya2015} gives that, if the cyber system  (\ref{EMMethod})-(\ref{def-X-app})  is mean-square exponentially stable, there is a Lyapunov function that proves the exponential stability of the  system. One may make use of this Lyapunov function for the cyber system and study the stability of  the physical system and that of the whole CPS.  Applying Theorem \ref{Theorem-MeanSquare}, we find that  the mean-square exponential stability of the CPS (\ref{SiDE-xy_app}) and, hence,  that of  the  physical system  (\ref{SDE_app}) can be inferred from the mean-square exponential stability of the cyber system (\ref{EMMethod})-(\ref{def-X-app})  if the Lyapunov function is of the quadratic form.

\begin{theorem}   \label{Theorem-MeanSquare-predict}
Assume that there is a candidate Lyapunov function $V \in C^2(\RR^n ; \RR_+)$  of the quadratic form (\ref{LyapunovFunction-MeanSquare})
 for the cyber system  (\ref{EMMethod})-(\ref{def-X-app}) with stepsize $\Delta t =\overline{\Delta t} >0$ such that
\begin{equation}   \label{Lyapunov-discrete-decreasing}
   \EE \big[ V (X_{k+1}) \big| X_k \big] \le \bar{c} \, V( X_k) 
\end{equation}
for some positive constant $\bar{c}<1$ and all $X_k \in \RR^n$.
Then  CPS  (\ref{SiDE-xy_app}) with  $\Delta t \in (0, \overline{\Delta t} ]$ is mean-square exponentially stable and    also almost surely exponentially stable, which implies that physical system  (\ref{SDE_app}) is  mean-square exponentially stable and    also almost surely exponentially stable.
\end{theorem}

\begin{proof} 
By Lyapunov stability theory   \cite{boyed1994,khas2012}, conditions (\ref{LyapunovFunction-MeanSquare}) and (\ref{Lyapunov-discrete-decreasing}) as well as (\ref{LyapunovFunctions_app}) derived from  (\ref{LyapunovFunction-MeanSquare}) immediately imply that the cyber system  (\ref{EMMethod})-(\ref{def-X-app}) with  $\Delta t = \overline{\Delta t} $ is mean-square exponentially stable. Let function $V(x)$ given by (\ref{LyapunovFunction-MeanSquare})  also be  the candidate Lyapunov function for the physical system (\ref{SDE_app}).  But condition (\ref{Lyapunov-discrete-decreasing}) 
\begin{multline*} 
  \EE \big[ V (X_{k+1}) \big| X_k \big] = \EE \big[ X^T_{k+1} P X^T_{k+1}  \big| X_k \big]     \\
  = \EE \Big[  \big( X_k +  f(X_k)  \overline{\Delta t}   + g( X_k) \sqrt{ \overline{\Delta t} } \, \xi (k+1) \big)^T  P    \big( X_k +  f(X_k)  \overline{\Delta t}   + g( X_k) \sqrt{ \overline{\Delta t} } \, \xi (k+1)  \big)  \big| X_k \Big]     \\
  = V( X_k )  +   \overline{\Delta t}  \Big[ X^T_k  P f(X_k)  + f^T (X_k) P X_k   + trace \big[ g^T(X_k ) P g (X_k )  \big]   +  \overline{\Delta t}  f^T (X_k) P f (X_k ) \Big]      \le  \bar{c} \, V( X_k )
\end{multline*}
implies that 
 $   X^T_k  P f(X_k)  + f^T (X_k) P X_k + trace \big[ g^T(X_k ) P g (X_k )  \big]    +  \overline{\Delta t}  f^T (X_k) P f (X_k )   \le - \bal V(X_k) $
for  all $X_k \in \RR^n$, where  $ \bal \,  {\overline{\Delta t}}  =  1- \bar{c}  $.
This means
\begin{equation*}
   \ll V(x ) + \overline{\Delta t} \, V(f(x)) 
  = x^T  P f(x)  + f^T (x) P x + trace \big[ g^T(x ) P g ( x )  \big]       +  \overline{\Delta t}  f^T (x) P f (x )     \le  - \bal V(x)   
\end{equation*}
for all $x \in \RR^n$, which is exactly the condition (\ref{LV-MeanSquare}) of Theorem \ref{Theorem-MeanSquare}.
According to Theorem \ref{Theorem-MeanSquare},    $z(t) = [ x^T(t) \; \, y^T(t)]^T$ of  CPS (\ref{SiDE-xy_app}) with  $\Delta t \in (0, \overline{\Delta t} ]$ is mean-square exponentially stable and is   also almost surely exponentially stable. It  immediately  follows that, due to $|x(t)|^2 \le |z(t)|^2$,   the physical  system  (\ref{SDE_app}) is  mean-square exponentially stable and is also almost surely exponentially stable. 
Alternatively,    conditions (\ref{LyapunovFunction-MeanSquare}) and
 $   \ll V(x ) \le \ll V(x ) + \overline{\Delta t} \, V(f(x)) \le   - \bal V(x) $
for all $x \in \RR^n$ as well as (\ref{LyapunovFunctions_app}) derived from  (\ref{LyapunovFunction-MeanSquare}) imply that,  by \cite[Theorems 4.2-4.4, pp128-130]{mao2007book},  the physical system (\ref{SDE_app}) is   mean-square exponentially stable and is also almost surely exponentially  stable.
\end{proof}

Our positive results to the key questions (Q1) and (Q2)   expose   the equivalence and intrinsic relationship $ \bal \,  {\overline{\Delta t}}  =  1- \bar{c}  $ between (\ref{LV-MeanSquare}) and (\ref{Lyapunov-discrete-decreasing}), which are the stability conditions for   the physical system (\ref{SDE_app}) and   its cybercounterpart  (\ref{EMMethod})-(\ref{def-X-app}), respectively. For this,  we employ the same Lyapuov function $V(x) = \tV (x) = x^T P x$ for both the   subsystems in Theorems  \ref{Theorem_StableContinuous_app}-\ref{Theorem-MeanSquare-predict}, see also Section \ref{sec:linearSDEs}.  Our proposed theory can further developed by using  some techniques of Lyapunov functions and functionals to exploit the  composition and structure  (\ref{linearSiDE-xy}), see Remark \ref{remark-upperbound} as well as \cite{fridman2014,huang2022,huang_partII,liu1988,nagh2008}.
It is also worth noting that the initial condition $y (0) = x(0) - X(0) =0$ is not a requirement in our established stability theory  and   its application in Section \ref{sec:linearSDEs}. But this condition could make a difference in the study of convergence and some control problems, see Appendix B and \cite{huang_partII}.

\section{Application to linear systems}   \label{sec:linearSDEs}

Let us consider a linear stochastic system
\begin{equation}    \label{linearSDE}
  \d x(t) = F x(t) \d t+ \sum_{j=1}^m G_j x(t) \d B_j (t)  \quad   \forall \, t \ge 0
\end{equation}
  with initial value $x(0)  \in \RR^{n}$, where $F \in \RR^{n \times n}$ and $ G_j \in \RR^{n \times n}$, $j=1, 2, \cdots, m$, are constant matrices. Obviously, the linear  system  (\ref{linearSDE}) satisfies the global Lipschitz   condition and  has a unique solution $x(t)$ on $  [0, \infty)$. It is well known that the linear  stochastic system (\ref{linearSDE}) is mean-square exponentially stable  if and only if there exists a positive definite matrix $P \in \RR^{n \times n}$ such that    \cite{boyed1994,elghaoui1995} 
\begin{equation}    \label{Lyapunov-Ito-LMI}
    F^T P + P F + \sum_{j=1}^m G_j^T P G_j < 0.   
\end{equation}
This is the classic Lyapunov-It{\^o} inequality \cite{boyed1994,elghaoui1995}, the linear matrix inequality (LMI)  equivalent \cite{gahinet1995} of the classical Lyapunov-It{\^o} equation  \cite{arnold1974, liu1992}.
By  \cite[Theorem 4.2, p128]{mao2007book}, the mean-square exponential stability of   (\ref{linearSDE}) implies that it is also almost surely exponentially stable.

The Euler-Maruyama method (\ref{EMMethod})-(\ref{def-X-app}) for the  linear  system  (\ref{linearSDE}) computes approximations
\begin{equation}     \label{linearSDE-EMMethod}
    X_{k+1} = X_k + F X_k \Delta t +  \sum_{j=1}^m G_j X_k \sqrt{\Delta t} \, \xi_j (k+1)  \quad   \forall \, k \in \NN
\end{equation}
 where $X_0 =x (0)$, $\Delta t>0$ is the constant stepsize and $\sqrt{\Delta t} \, \xi_j (k+1)$ is the implementation of $\Delta B_{j, k} = B_j ( (k+1) \Delta t) -B_j (k \Delta t) $. The cyber system (\ref{linearSDE-EMMethod}) is mean-square exponentially stable if and only if there exists a positive definite matrix $P \in \RR^{n \times n}$ such that \cite{boyed1994}
\begin{equation}     \label{Lyapunov-discrete-LMI}
   (I_n + \Delta t \, F)^T P  ( I_n + \Delta t \, F)  +  \Delta t \, \sum_{j=1}^m G_j^T P G_j <  P.
\end{equation}

Let $y(t)$ be the   difference between $x(t)$ and $X(t)$ as (\ref{def-y-app}) above. The  CPS of the Euler-Maruyama method (\ref{linearSDE-EMMethod}) for the linear SDE (\ref{linearSDE}) is  a specific   case of  CPS  (\ref{SiDE-xy_app})
\begin{subequations}  \label{linearSiDE-xy} 
\begin{align}
& \mathrm{d}  x(t) = Fx(t)  \mathrm{d}t + \sum_{j=1}^m G_j x(t) \mathrm{d} B_j (t)   \label{linearSDE_x} \\
& \d y(t)  = Fx(t) \mathrm{d}t + \sum_{j=1}^m G_j x(t) \mathrm{d} B_j (t)  \quad   t \in   [t_k, t_{k+1})  \label{linearSDE_y}    \\
& \tilde{\Delta} (x(t_{k+1}^-), y(t_{k+1}^-), k+1) := y(t_{k+1}) - y(t_{k+1}^-)  \nonumber  \\
 & \quad  {}       = - F \big( x(t_{k+1}^-)-y(t_{k+1}^-) \big) \Delta t    - \sum_{j=1}^m G_j \big( x(t_{k+1}^-)- y(t_{k+1}^-) \big) \sqrt{\Delta t} \, \xi_j (k+1)  \quad k \in \NN
\label{linearImpulse_y} 
\end{align}
\end{subequations}
with initial data $x(0) \in \RR^n$ and $ y(0)=   x(0)-X(0)  =0 $, where $t_k = k  \Delta t$ for all $k \in \NN$. The CPS (\ref{linearSiDE-xy}) is an integration of the physical  system (\ref{linearSDE}) and the cyber system  (\ref{linearSDE-EMMethod}), which is of our proposed class  (\ref{SiDE-xy})  and satisfies the global Lipschitz conditions Assumption  \ref{globalLipschitz}.    
Our established theory  immediately provides  positive results  to  the key questions  (Q1) and (Q2) of linear CPS  (\ref{linearSiDE-xy}), which presents the upper bound $\overline{\Delta t}$ of stepsizes for exponential stability.

\begin{theorem}  \label{linearSDE-equivalence}
The following are equivalent.

\begin{itemize}
\item[($\mathcal{A}$)]  
  There exists a  positive definite matrix $P \in \RR^{n \times n}$   such that the CPS Lyapunov inequality holds for some positive number  $\overline{\Delta t}$, namely,
\begin{equation}    \label{numerical-Lyapunov-Ito-LMI}
    F^T P + P F + \sum_{j=1}^m G_j^T P G_j  + \overline{\Delta t} F^T P F < 0.
\end{equation}

\item[($\mathcal{B}$)] 
The physical system (\ref{linearSDE}) is mean-square exponentially stable. 

\item[($\mathcal{C}$)] 
 The cyber system (\ref{linearSDE-EMMethod}) with  $\Delta t \in (0, \overline{\Delta t} \,]$ is mean-square exponentially stable. 

\item[($\mathcal{D}$)] 
 The CPS  (\ref{linearSiDE-xy}) with  $\Delta t \in (0, \overline{\Delta t} \,]$  is mean-square exponentially stable.
\end{itemize}

That is,   $ (\mathcal{A}) \Leftrightarrow  (\mathcal{B})  \Leftrightarrow  (\mathcal{C}) \Leftrightarrow  (\mathcal{D})$.
\end{theorem}

\begin{proof}   $  (\mathcal{A}) \Leftrightarrow  (\mathcal{B})  $.  We only need to show that the classic Lyapunov inequality  (\ref{Lyapunov-Ito-LMI}) and the CPS Lyapunov inequality (\ref{numerical-Lyapunov-Ito-LMI}) are equivalent. But  this is implied by the equivalence of the   inequalities  (\ref{LV_app-x}) and  
(\ref{LV-MeanSquare}) which has been shown in the proof of Theorem \ref{Theorem-MeanSquare}. Alternatively, it can be easily proved  as follows.
Clearly, (\ref{numerical-Lyapunov-Ito-LMI}) implies (\ref{Lyapunov-Ito-LMI}). But inequality  (\ref{Lyapunov-Ito-LMI}) implies that there is a sufficiently small positive number $ \overline{\Delta t}$ such that  (\ref{numerical-Lyapunov-Ito-LMI}) holds. Therefore,  the LMI (\ref{Lyapunov-Ito-LMI}) $  \Leftrightarrow $ the LMI (\ref{numerical-Lyapunov-Ito-LMI}).

 $  (\mathcal{A})   \Rightarrow (\mathcal{C}) \;  \&  \;  (\mathcal{D}) $. Let us consider the quadratic Lyapunov function
   $ V(x) = x^T P x $ 
for the linear system (\ref{linearSDE}). The LMI (\ref{numerical-Lyapunov-Ito-LMI}) implies that there is a positve number $\bal < 1/ \, \overline{\Delta t}$ sufficiently small for 
\begin{equation}    \label{Lyapunov-Ito-LMI_y-bal}
    F^T P + P F + \sum_{j=1}^m G_j^T P G_j  + \overline{\Delta t} F^T P F \le  - \bal P,    
\end{equation}
and the condition (\ref{LV-MeanSquare})   
holds.
 It follows from  Theorem \ref{Theorem-MeanSquare} that the CPS  (\ref{linearSiDE-xy}) and, hence, the cyber system  (\ref{linearSDE-EMMethod}) with  $\Delta t \in (0, \overline{\Delta t} \,] $ are mean-square exponentially stable.

$(\mathcal{D}) \Rightarrow  (\mathcal{B}) \; \& \;  (\mathcal{C})  $.   The CPS  (\ref{linearSiDE-xy}) is a specific case of   system (\ref{Compact-z}), where $z(t) = [ x^T(t) \; \, y^T(t) ]^T$ in the compact form. Notice that 
 $|x(t)|^2 \le | z(t) |^2 $ and $ |X(t)|^2 \le 2 (|x (t)|^2 + |y (t)|^2) \le 4 |z(t) |^2 $ for all $t \ge 0$. If $z(t) $  of the CPS  (\ref{linearSiDE-xy}) is  mean-square exponentially stable, then both $x(t)$ of its physical subsystem (\ref{linearSDE}) and $X(t)$ of its cyber subsystem (\ref{linearSDE-EMMethod})  are mean-square exponentially stable.   

$ (\mathcal{C}) \Rightarrow (\mathcal{B}) \;  \& \; (\mathcal{D})  $. Let $\Delta t = \overline{\Delta t} $. Since the cyber  system  (\ref{linearSDE-EMMethod}) is mean-square exponentially stable, there is a positive definite matrix $P \in \RR^{n \times n}$ such that the Lyapunov inequality (\ref{Lyapunov-discrete-LMI}) holds for $\Delta t = \overline{\Delta t} >0$. This implies that there is a positive number $\bar{c} \in (0, 1) $   sufficiently close to $1$ for
\begin{equation}     \label{Lyapunov-discrete-LMI-1}
    ( I_n + \overline{\Delta t} \, F)^T P  ( I_n +  \overline{\Delta t} \, F)  +  \overline{\Delta t} \, \sum_{j=1}^m G_j^T P G_j \le \bar{c} \, P.
\end{equation}
Let the quadratic  function 
    $ V(x) = x^T P x  $ 
be the candidate Lyapunov function for the cyber system (\ref{linearSDE-EMMethod}) with $\Delta t = \overline{\Delta t} $. Observe that, for the linear system,  (\ref{Lyapunov-discrete-LMI-1}) is exactly the condition (\ref{Lyapunov-discrete-decreasing})  of Theorem \ref{Theorem-MeanSquare-predict}. It follows from Theorem \ref{Theorem-MeanSquare-predict} that the CPS  (\ref{SiDE-xy_app}) with  $\Delta t \in (0, \overline{\Delta t} ]$ and, hence, the physical  system   (\ref{SDE_app}) are  mean-square exponentially stable.
The proof is complete.
\end{proof}

Note that the mean-square exponential stability of the   physical system (\ref{linearSDE}), the cyber system (\ref{linearSDE-EMMethod}) and the CPS  (\ref{linearSiDE-xy})  imply that they are also almost surely exponentially stable, respectively. 
It is easy to obtain the upper bound $\overline{\Delta t}$ of stepsizes for the ability of the cyber system to reproduce the exponential stability of the underlying linear  physical system by  solving the CPS Lyapunov  inequality   (\ref{numerical-Lyapunov-Ito-LMI}), which can also be called the numerical Lyapunov  inequality in  the study of numerical methods for differential equations. 

\begin{remark}    \label{CPS-Lyapunov-scalar}
Consider a scalar  SDE,  which is the linear SDE (\ref{linearSDE}) with $n=m=1$, 
\begin{equation}   \label{scalarSDE-linear}
   \d x(t) = \lambda x(t) \d t + \mu x(t) \d B(t), \quad t \ge 0, \quad x(0)   \neq 0
\end{equation}
where $\lmd$ and $\mu$ are both real constants. 
The CPS Lyapunov inequality (\ref{numerical-Lyapunov-Ito-LMI}) immediately gives 
\begin{equation}   \label{Lyapunov-Ito-LMI-scalar}
   2 \lmd + \mu^2 + \lmd^2 \,  \overline{\Delta t} <0 \quad \Leftrightarrow \quad \overline{\Delta t} <  \frac{ - (2 \lmd + \mu^2 )} {\lmd^2}.
\end{equation}
According to Theorem \ref{linearSDE-equivalence},  this is the necessary and sufficent condition for mean-square exponential stability of    the linear scalar physical system (\ref{scalarSDE-linear}), its cyber system of the Euler-Maruyama method, and its CPS of the Euler-Maruyama method, that is, (\ref{linearSiDE-xy}) with   $n=m=1$, $F= \lambda$ and $G_1 = \mu$.
It is observed that  (\ref{Lyapunov-Ito-LMI-scalar}) is exactly the inequality (4.3) in \cite{higham2000} with $\theta =0$ for the Euler-Maruyama method.
Obviously,  (\ref{Lyapunov-Ito-LMI-scalar})   is the very   scalar case of our CPS Lyapunov inequality  (\ref{numerical-Lyapunov-Ito-LMI})
that applies to general multi-dimensional linear systems. 
\end{remark}

Recently,   based on the reformulation of some well-known results, 
\cite{buckwar2012a} developed an approach to mean-square  stability analysis of numerical methods (including  the widely-used  Euler-Maruyama scheme) for multi-dimensional linear SDEs  (viz. system  (\ref{linearSDE}) with $n \ge 2$), which was   applied in \cite{buckwar2012b}  to study the mean-square numerical stability for a linear SDE of non-normal drift and skew-symmetic diffusion structures \cite{higham2006}.  Specifically, on one hand, some well-known result (\cite[Theorem 8.5.5, p142]{arnold1974}, \cite[Remark 6.4, p183]{khas2012}) expressed in the vectorization of matrices and Kronecker product gives \cite[Lemma 3.3]{buckwar2012a}
\begin{equation}    \label{MS-SDE}
  Re_M ({\mathcal S}) < 0
\end{equation}
if and only if linear SDE  (\ref{linearSDE}) is mean-square exponentially stable, where $Re_M ({\mathcal S})$ is the real part of the eigenvalue $\lambda_M ({\mathcal S})$ of $n^2 \times n^2$ matrix 
$$
 {\mathcal S} := I_n \otimes F + F \otimes I_n + \sum_{j=1}^m G_j \otimes G_j .
$$
On the other hand, a stability result for discrete-time stochastic systems (see, e.g., \cite[p197]{khas2012}) is applied to study mean-square stability of some numerical schemes for the SDE  (\ref{linearSDE}). For example, by  \cite[Lemma 3.4, Theorem 3.7]{buckwar2012a},  the Euler-Maruyama method (\ref{linearSDE-EMMethod}) is mean-square exponentially stable if and only if
\begin{equation}    \label{MS-SDE-EM}
    \rho ({\mathcal S}_0 (\Delta t) ) <1,
\end{equation}
where $ \rho ({\mathcal S}_0 (\Delta t ) )$ is the spectral radius of  $n^2 \times n^2$ matrix
$$
   {\mathcal S}_0 ( \Delta t) := \big( \bA (\Delta t)  \otimes \bA (\Delta t) \big)  + \sum_{j=1}^m \big( \bB_j (\Delta t) \otimes \bB_j  (\Delta t)   \big) 
$$ 
with $  \bA (\Delta t) = I_n + {\Delta t} F $ and $  \bB_j ( \Delta t) = \sqrt{ \Delta t } \, G_j$ for $j =1, \cdots, m$.

In \cite{buckwar2012a},   ${\mathcal S}$ is called the mean-square stability matrix  of the SDE (\ref{linearSDE}) and $ {\mathcal S}_0 ( \Delta t)$  that of the Euler-Maruyama method (\ref{linearSDE-EMMethod}).
Notice that  $ {\mathcal S}_0 ({\Delta t}) $ is a function of stepsize ${\Delta t}$ while, obviously,   ${\mathcal S}$ is not. 
The results in  \cite{buckwar2012a,buckwar2012b} provided the  explicit structure of stability matrices ${\mathcal S}$ and $ {\mathcal S}_0 ({\Delta t}) $,  and showed the comparative stability regions \cite[Fig.2]{buckwar2012b} for the   SDE and the   numerical method with a few numerical examples of non-normal SDEs  \cite{higham2006}.
However, the relationship between the stability conditions (\ref{MS-SDE}) and  (\ref{MS-SDE-EM})  (for the  SDE and the   Euler-Maruyama method, respectively)  has not been figured out. 
Here we prove their equivalence and reformulate the stability conditions (\ref{MS-SDE}) and (\ref{MS-SDE-EM}) in the form of LMIs, which   is relegated to Appendix A. So it is easy to handle the problems with some computing techniques  and toolboxes \cite{boyed1994,elghaoui1995,gahinet1995}.   

Note that  it is easy to obtain the upper bound $\overline{\Delta t}$ of stepsizes for the test problem (Q1) by solving the $n \times n$-dimensional LMI  (\ref{numerical-Lyapunov-Ito-LMI})  of our proposed method, which is a linear function with respect to  $\Delta t$. Clearly, LMI  (\ref{numerical-Lyapunov-Ito-LMI}) holds for all $\Delta t \in (0, \overline{\Delta t}]$ provided it is satisfied for  some $ \overline{\Delta t}>0$.
But, to calculate the upper bound $ \tDelt$ by the approach  of mean-square stability matrix \cite{buckwar2012a}, one has to deal with        the spectral radius (\ref{MS-SDE-EM})  of    $n^2 \times n^2$ matrix $ {\mathcal S}_0 ({\Delta t}) $ that involves a polynomial of the stepsize ${\Delta t}$ whose order is some exponential function of $n$. Alternatively, one can  solve  the following  LMI  with respect to positive definite matrix $\bP \in \RR^{n^2 \times n^2}$, which we show is equivalent to inequality  (\ref{MS-SDE-EM})   in  Appendix A,
\begin{multline}   \label{LMI-2012}
  {\mathcal S}^T \bP + \bP   {\mathcal S} +   {\Delta t}     \big(  {\mathcal S}^T \bP  {\mathcal S} + \bF^T \bP + \bP \bF \big)     + ( {\Delta t} )^2 \big(  {\mathcal S}^T \bP  \bF + \bF^T \bP {\mathcal S}  \big)  + ( {\Delta t} )^3  \bF^T \bP \bF  \\
=  \big(  {\mathcal S}   +  {\Delta t}  \bF    \big)^T  \bP + \bP \big(  {\mathcal S}     +  {\Delta t}   \bF \big)      +   {\Delta t}    \, \big(  {\mathcal S}   +  {\Delta t}  \bF    \big)^T  \bP \big(  {\mathcal S}     +  {\Delta t} \bF \big) 
< 0, 
\end{multline}  
where $ \bF= F \otimes F $. This involves a cubic function of the prescribed parameter $\Delta t >0$ for all $n$. Note that, for a multi-dimensional SDE ($n \ge 2$),  the spectral radius (\ref{MS-SDE-EM}) of  $n^2 \times n^2$ matrix $ {\mathcal S}_0 ({\Delta t}) $  involves a polynomial of $\Delta t$ of up to (very) high order. For example, in the case   $n=2$ of non-normal SDE 
 \cite[Eq.(9)]{buckwar2012b}  (see also \cite{higham2006}),  the characteristic equation of mean-square stability matrix $ {\mathcal S}_0 ( \Delta t)$ of $4 \times 4$ dimensions for the Euler-Maruyama scheme (\cite[Eq.(15)]{buckwar2012b} with $\theta =0$)  involves a polynomial of $\Delta t$ of up to order $8$. The conditions of this approach  are quite cumbersome \cite{khas2012}. It is easy to tackle the equivalent LMI (\ref{LMI-2012}), which is a cubic function of the prescribed $\Delta t >0$ for all $n$,  with some toolbox  \cite{boyed1994,gahinet1995}. 
However, one should be aware  that, unlike the linear inequality (\ref{numerical-Lyapunov-Ito-LMI}), that  LMI (\ref{LMI-2012}), which is not only of  $n^2 \times n^2$   dimensions but also involves a cubic function of the prescribed  $\Delta t >0$,  is satisfied for some $\oDt >0$ does not necessarily mean that it holds for all $\Delta t \in (0,   \oDt ]$ and, thereby, the results such as the upper bound $ \oDt$ could be  restrictive due to  the highly nonlinearity of $\Delta t$ involved in the computation, see Appendix A.

We  can further show that our proposed  method  (\ref{numerical-Lyapunov-Ito-LMI}) gives better bound $\oDt$  of  stepsizes than  (\ref{MS-SDE-EM}) from  \cite{buckwar2012a} or its LMI equivalent  (\ref{LMI-2012}). This is: if (\ref{MS-SDE-EM}) and  (\ref{LMI-2012})   hold for all  $ \Delta t \in (0, \tDelt ] $ for some $\tDelt>0$, then the CPS Lyapunov inequality  (\ref{numerical-Lyapunov-Ito-LMI}) holds for some $ \oDt \ge \tDelt $, namely,  either $ \oDt = \tDelt $ or  $ \oDt > \tDelt $.  In short, we shall show either $  \hDelt = \tDelt $ or  $  \hDelt > \tDelt $, where 
\begin{equation}   \label{comparisonLMIs}
\tDelt := \sup \{ \oDt  >0:  {\rm  (\ref{LMI-2012})  \; holds \; for \; all \;} \Delta t \in (0, \oDt \, ] \} \; {\rm and} \;  \hDelt := \sup \{ \oDt >0:  {\rm  (\ref{numerical-Lyapunov-Ito-LMI}) \; holds }   \}. 
\end{equation}
It is observed that,  due to the continuity of (\ref{LMI-2012}) with respect to $\Delta t$, the strict inequality (\ref{LMI-2012}) does not hold at   $\Delta t = \tDelt$ and, similarly,   the strict inequality (\ref{numerical-Lyapunov-Ito-LMI}) does not hold at $\oDt = \hDelt$. 
To show either $  \hDelt = \tDelt $ or  $  \hDelt > \tDelt $,    let us   consider a linear SDE with $ \Delta t \in (0, \oDt \, ]$ for some $\oDt >0 $
\begin{equation}    \label{linearSDE-extra}
   \d x(t) = F x(t) \d t+ \sum_{j=1}^m G_j x(t) \d B_j (t)   + \sqrt{ \Delta t  }\, F x(t) \d B_{m+1} (t)   \quad \forall \, t \ge 0
\end{equation}
  where     $B_{m+1} (t)$ is a scalar Brownian motion. Notice that (\ref{linearSDE-extra}) is exctly   (\ref{linearSDE}) if $\Delta t =0$ and, according to  \cite[Lemma 3.3]{buckwar2012a}, the SDE (\ref{linearSDE-extra}) is  mean-square exponentially stable if and only if (\ref{MS-SDE-t}) holds.
 In fact, by the well-known results  \cite{arnold1974,boyed1994,buckwar2012a,gahinet1995,khas2012}, the following are equivalent: 
 \begin{itemize}

 \item[(a)] The CPS Lyapunov  LMI   (\ref{numerical-Lyapunov-Ito-LMI}) holds. 

 \item[(b)]  There is a positive definite matrix $\bP   \in \RR^{n^2 \times n^2}$ such that
\begin{equation}      \label{EM-bP-4}
  \big(  {\mathcal S}   +  {\Delta t}  \bF    \big)^T  \bP + \bP \big(  {\mathcal S}     +  {\Delta t}   \bF \big)  = {\mathcal S}^T \bP + \bP   {\mathcal S}  
     +  {\Delta t}     \big(   \bF^T \bP + \bP \bF \big)  <  0  
\end{equation}
 for each $ {\Delta t} \in (0, \oDt \, ]$.

 \item[(c)] The following inequality holds for each ${\Delta t} \in (0, \oDt \, ]$
\begin{equation}    \label{MS-SDE-t}
Re_M ({\mathcal S}  +  {\Delta t}  \bF ) < 0.
\end{equation}

  \item[(d)] The SDE (\ref{linearSDE-extra})    is  mean-square exponentially stable. 
\end{itemize}
That is, (a) $ \Leftrightarrow$ (b) $ \Leftrightarrow$ (c) $ \Leftrightarrow$ (d) and they all give the same  supremum $ \hDelt$ defined as (\ref{comparisonLMIs}). 
 Observe that (\ref{LMI-2012})  implies   (\ref{EM-bP-4}) but not vice versa. This means that (\ref{numerical-Lyapunov-Ito-LMI}),   (\ref{EM-bP-4}) and  (\ref{MS-SDE-t}) hold for all $ \oDt \in (0, \tDelt )$, namely, $\hDelt \ge \tDelt$. Notice that (\ref{LMI-2012}), (\ref{EM-bP-4}) and  (\ref{MS-SDE-t}) are all continuous  with respect to $\Delta t$. 
Due to the continuity of  (\ref{LMI-2012}) at $\Delta  t= \tDelt $,  
\begin{multline*}
 \big(  {\mathcal S}   +  \tDelt   \bF    \big)^T  \bP + \bP \big(  {\mathcal S}     +  \tDelt    \bF \big)  +  \tDelt     \, \big(  {\mathcal S}   +  \tDelt  \bF    \big)^T  \bP \big(  {\mathcal S}     +  \tDelt   \bF \big)  \le 0    \\
\Leftrightarrow  \quad   \big(  {\mathcal S}   +  \tDelt   \bF    \big)^T  \bP + \bP \big(  {\mathcal S}     +  \tDelt    \bF \big)  \le -  \tDelt     \, \big(  {\mathcal S}   +  \tDelt  \bF    \big)^T  \bP \big(  {\mathcal S}     +  \tDelt   \bF \big) .
\end{multline*} 
Unless matrix $ {\mathcal S}     +  \tDelt   \bF $ is singular,  the LMIs (\ref{EM-bP-4}) and, equivalently, (\ref{numerical-Lyapunov-Ito-LMI}) hold at   $\oDt = \tDelt$ and, due to their continuity   at $\Delta t = \tDelt $, the LMIs (\ref{EM-bP-4}) and  (\ref{numerical-Lyapunov-Ito-LMI}) hold for some $\oDt >  \tDelt $, which gives  $\hDelt > \tDelt$. 
So we have  $\hDelt = \tDelt$ if matrix $ {\mathcal S}     +  \tDelt   \bF $ is singular; otherwise, $\hDelt > \tDelt$.  The latter could often be   the case. This clearly shows that our proposed CPS Lyapunov LMI  (\ref{numerical-Lyapunov-Ito-LMI}) gives better bound $\oDt$  of stepsizes than  (\ref{MS-SDE-EM}) from  \cite{buckwar2012a}, or, its  LMI equivalent (\ref{LMI-2012}).

\begin{remark}   \label{remark-upperbound}
The upper bound $\oDt$  of stepsizes in the CPS Lyapunov inequality (\ref{numerical-Lyapunov-Ito-LMI}) is obtained  by the same Lyapunov function $V(x) = \tV (x) = x^T P x$ for both the   subsystems,  which is special application of  Theorem \ref{Theorem_StableContinuous_app}.     The upper bound $\oDt$  could be improved by using some other Lypunov functions or functionals  such as different Lyapunov functions for the subsystems \cite{huang_partII} and some  Lyapunov function or funcational of quadratic form with block matrices for the compact form of the CPS \cite{fridman2014,huang2022,liu1988,nagh2008} to exploit the structure of the composition (\ref{linearSiDE-xy})  of the subsystems \cite{liu1988}. 
\end{remark}

\section{On systems numerics}   \label{sec:conclusion}

In this paper,  we formulated a new and general class (\ref{SiDE-xy}) of SiDEs that can be used to represent a seamless integration of the physical system (the SDE) and its cyber counterpart  (the numerical method), which is  a novel framework for   numerical study of dynamical systems.  Our proposed CPS of the widely-used Euler-Maruyama method for SDEs  not only provides a holistic view of the physical system (the SDE) and its cyber counterpart  (the numerical method) but also reveals their intrinsic relationship: they are not two separate systems but the subsystems of the CPS.  By our CPS approach, we proved positive results to the key questions (Q1) and (Q2) using the Lyapunov stability theory we established for our general class of SiDEs. These fundamental results and their application  comprise a foundational theory of the CPSs of numerical methods for SDEs. 
   
In the classical numerical analysis of initial-value problems, the convergence and the stability of a numerical method are two main concerns. The proposed CPS  also provides a novel approach to convergence
analysis of the numerical method for SDEs.
Using the dyanimcs of the discretization error $y(t)$,  we show the classical finite-time convergence result 
$$
  \EE \,\Big[ \sup_{0\le t \le T}  | y(t) |^2 \Big]  = O ({\Delta t})
$$
 for the Euler-Maruyama method by our CPS approach, which is different from those   in the literature \cite{higham2002,kloeden1992,mao2007book,stuart1996}. The novel proof is relegated to Appendix B.

Our proposed CPS and theory   have initiated the study of systems numerics, where there are many open and interesting problems. 
For example, it is among the future work to extend our   theory to many other (explicit or implicit) numerical methods \cite{higham2000,higham2002,hutzenthaler2012,kloeden1992,mao2015b,sabanis2013,stuart1996} as well as to  dyanmical systems such as   SDEs  with time delay and/or switching \cite{huang2009,huang2022}.

\section*{Acknowledgement}
The author is supported  by  National Natural Science Foundation of China (No.61877012).   
The author gratefully acknowledges Prof. X. Mao  (the author's PhD supervisor at University of Strathclyde, UK) for his very helpful comments and suggestions, which   improve the quality of this work. 
Some 
of this work was done during the author's visit to The Centre for Stochastic \& Scientific Computations at Harbin Institute of Technology and  the author would like to thank Prof. M. Song, Prof. M. Liu and their group for the helpful discusssions. The author is also thankful to Prof. K. C. Sou for  the communications  in the past a few years.  
 
\section*{Appendix A.  The equivalence of the stability conditions (\ref{MS-SDE}) and (\ref{MS-SDE-EM})}    \label{sec:App-A}
\begin{proof} It is observed that
 $ {\mathcal S}_0 ( 0)= I_{n^2} $,  $ \dot{ {\mathcal S}}_0 ( 0)=  {\mathcal S} $, 
 $\ddot{ {\mathcal S}}_0 ( \Delta t) = 2 \, ( F \otimes F )   =: 2 \bF $,
where $\dot{{\mathcal S}}_0  ( \Delta t)$ and $\ddot{{\mathcal S}}_0  ( \Delta t)$ are the first and the second derivatives of ${\mathcal S}_0$ with respect to $\Delta t$, respectively.  For   $\Delta t >0$, Taylor expansion produces
\begin{equation}  \label{S0-TaylorSeries}
          {\mathcal S}_0 ( \Delta t) = I_{n^2} + \Delta t \, {\mathcal S} + (\Delta t)^2  \bF .
\end{equation} 

   (\ref{MS-SDE}) $\Rightarrow$   (\ref{MS-SDE-EM}).   Stability condition (\ref{MS-SDE}) for the linear SDE equivalently means that there is a positive definite matrix $\bP \in \RR^{n^2 \times n^2}$ such that   \cite{boyed1994,khas2012} 
\begin{equation}  \label{SDE-bP}
        {\mathcal S}^T \bP + \bP   {\mathcal S} < 0.
\end{equation}
The   Taylor series (\ref{S0-TaylorSeries}) gives
\begin{multline}  \label{EM-bP-0}
    {\mathcal S}_0^T ( \Delta t) \bP  {\mathcal S}_0 ( \Delta t)     
   = \bP + {\Delta t}    \big[  {\mathcal S}^T \bP + \bP   {\mathcal S}  
     +  {\Delta t}     \big(  {\mathcal S}^T \bP  {\mathcal S} + \bF^T \bP + \bP \bF \big)    \\   {} + ( {\Delta t} )^2 \big(  {\mathcal S}^T \bP  \bF + \bF^T \bP {\mathcal S}  \big)  + ( {\Delta t} )^3  \bF^T \bP  \bF \big].
\end{multline}
Owing to   (\ref{SDE-bP}),  there is a pair of (sufficiently small) positive numbers $\oDt $ and $\ba = \ba (\oDt )$ such that
\begin{equation*}   
    {\mathcal S}_0^T ( \Delta t) \bP  {\mathcal S}_0 ( \Delta t) \le \bP + \ba \, {\Delta t} \big(  {\mathcal S}^T \bP + \bP   {\mathcal S}  \big) < \bP    
\quad  \forall  \,   {\Delta t} \in (0,  \oDt  \,] .
\end{equation*}
 This implies that stability condition  (\ref{MS-SDE-EM}) 
 is satisfied 
 for all  $ {\Delta t} \in (0,  \oDt \,]$.

  (\ref{MS-SDE-EM}) $\Rightarrow$    (\ref{MS-SDE}). 
 Notice that  (\ref{EM-bP-0}) can be rewritten as
\begin{multline} 
 {\mathcal S}_0^T ( \Delta t) \bP  {\mathcal S}_0 ( \Delta t)     
   = \bP + {\Delta t}    \big[  {\mathcal S}^T \bP + \bP   {\mathcal S}  
     +  {\Delta t}     \big(  {\mathcal S}^T \bP  {\mathcal S} + \bF^T \bP + \bP \bF \big)    \\   {} + ( {\Delta t} )^2 \big(  {\mathcal S}^T \bP  \bF + \bF^T \bP {\mathcal S}  \big)  + ( {\Delta t} )^3  \bF^T \bP  \bF \big] \\
  = \bP + {\Delta t}   \Big[  \big(  {\mathcal S}   +  {\Delta t}  \bF    \big)^T  \bP  + \bP \big(  {\mathcal S} +  {\Delta t}   \bF \big)          +   {\Delta t}    \, \big(  {\mathcal S}   +  {\Delta t}  \bF    \big)^T  \bP \big(  {\mathcal S}     +  {\Delta t} \bF \big)  \Big] .   \label{EM-bP-2}
\end{multline}
%
Suppose that condition  (\ref{MS-SDE-EM}) holds for all $ {\Delta t} \in (0, \oDt ]$, where  $\oDt$ is some  positive number. Then, equivalently,  there is a positive definite matrix $ \bP = \bP({\Delta t}) \in \RR^{n^2 \times n^2}$  such that    \cite{boyed1994,khas2012} 
\begin{equation}  \label{EM-bP-3}
    {\mathcal S}_0^T ( {\Delta t}  ) \bP  {\mathcal S}_0 ({\Delta t}  )  <  \bP  \quad \forall \,  {\Delta t} \in (0, \oDt \, ].
\end{equation}
Substitution of (\ref{EM-bP-2}) into (\ref{EM-bP-3}) produces  (\ref{LMI-2012}) for all ${\Delta t} \in (0, \oDt \, ]$, 
or, by Schur complement, 
\begin{equation}     \label{EM-bP-3-2}
    \begin{bmatrix} 
    \big(  {\mathcal S}   +  {\Delta t}  \bF    \big)^T  \bP + \bP \big(  {\mathcal S}     +  {\Delta t}   \bF \big)  
&    \sqrt{\Delta t}    \, \big(  {\mathcal S}   +  {\Delta t}  \bF    \big)^T  \bP   \\
       \sqrt{\Delta t}    \, \big(  {\mathcal S}   +  {\Delta t}  \bF    \big)^T   \bP   &   -  \bP    \end{bmatrix}      < 0   \quad \forall \, {\Delta t} \in (0, \oDt \, ].
\end{equation}
So  
$  (\ref{MS-SDE-EM})  \, \Leftrightarrow \,  (\ref{LMI-2012}) \, \Leftrightarrow \,  (\ref{EM-bP-3})  \,  \Leftrightarrow \,     (\ref{EM-bP-3-2}) $.
But   (\ref{LMI-2012}) implies the LMI (\ref{EM-bP-4}).
 This equivalently means that matrix ${\mathcal S}   +  {\Delta t}  \bF$ is Hurwitz, namely, inequality (\ref{MS-SDE-t}) holds for each $   {\Delta t} \in (0, \oDt \, ]$. 

Recall that $\bF =  F \otimes F $.  
Letting ${\Delta t} \to 0$ in  (\ref{EM-bP-4}) and thus   (\ref{MS-SDE-t})  gives stability condition (\ref{MS-SDE}) for the SDE (\ref{linearSDE}). The proof is complete.
\end{proof}

We remark that,  by  approach of mean-square stability matrices \cite{buckwar2012a}, the upper bound $\oDt$ can be  calculated by solving either   the spectral radius problem (\ref{MS-SDE-EM}) or   
 the LMI  equivalent (\ref{LMI-2012}). 
%
The former  involves a  polynomial of the stepsize $\Delta t$ whose order is some exponential function of $n$  while the latter remains as a cubic function of  $\Delta t$ for all $n$.   
The highly nonlinearity would introduce not only computational complexity  but also conservativeness to the results. We have reformulted the highly nonlinear problem (\ref{MS-SDE-EM}) into the  LMI (\ref{LMI-2012}). This has significantly simplified the approach of mean-square stability matrices ${\mathcal S}$ and $ {\mathcal S}_0 ( \Delta t)$.  But it is  worth noting that, for a linear $n$-dimensional SDE, our proposed  numerical Lyapunov  LMI  (\ref{numerical-Lyapunov-Ito-LMI}) of $n \times n$ dimensions is always a linear inequality of the stepsize $ \oDt$ while the LMI problem (\ref{LMI-2012})  involves not only a cubic function of $ \Delta t$ but also matrices of $n^2 \times n^2 $ dimensions.

\section*{Appendix B.  A novel proof of the convergence of  the Euler-Maruyama method}   \label{sec:App-B}
\begin{proof} 
For the   convergence problem of the numerical method, the implimentation  $ \sqrt{\Delta t} \, \xi (k+1) $  should be replaced by the increment $\Delta B_k = B( (k+1) \Delta t) -B(k \Delta t)$ itself in  SiDE (\ref{SiDE-xy_app}), that   is, 
\begin{subequations}  \label{SiDE-xy_conv}
\begin{align}
& \mathrm{d}  x(t) = f (x(t) )\mathrm{d}t + g(x(t)  )\mathrm{d} B(t)   \label{SDE_x_conv} \\
& \d y(t)  = f(x(t) )\mathrm{d}t + g(x(t)  )\mathrm{d} B(t)
\quad t \in [t_k, t_{k+1})  \label{SDE_y_conv}      \\    
&\tilde{\Delta} (x(t_{k+1}^-), y(t_{k+1}^-), k+1) := y(t_{k+1}) - y(t_{k+1}^-)  \nonumber  \\
 &  \qquad    = - f \big( x(t_{k+1}^-)-y(t_{k+1}^-) \big) \Delta t    - g \big( x(t_{k+1}^-)- y(t_{k+1}^-)  \big) \Delta B_k  \quad  k \in \NN
\label{Impulse_y_conv} 
\end{align}
\end{subequations}
with $x(0) \in \RR^n$ and $y(0)=0$, where $t_k =  k \Delta t$ for all $k \in \NN$. According to the existing results (\cite{mao2007book,kloeden1992} as  well as Lemma \ref{existence_n_uniqueness}), SiDE (\ref{SiDE-xy_conv}) has a unique (right-continuous) solution $z(t)=[ x^T (t)  \;  y^T(t) ]^T$, which belongs to  $\mm^2 ([0, T]; \RR^{n+q})$ for all $T \ge 0$. 
In particular,  \cite[Lemma 3.2, p51]{mao2007book} gives
\begin{equation}    \label{conv-Esupx}
  \EE \big[  \sup_{0 \le t \le T} | x (t) |^2  \big] \le (1 + 3|x(0)|^2) e^{ 3L T (T+4)} =: C_T   . 
\end{equation}
On the interval $[ t_k, t_{k+1}]$ for every $k \in \NN$, 
\begin{multline*}    
      y(t_{k+1}) - y(t_k) =    \int_{t_k}^{t_{k+1}^-} f(x(t)) \mathrm{d}t    + \int_{t_k}^{t_{k+1}^-} g(x(t)) \mathrm{d} B(t) \\
{}   - f (x(t_{k+1}^-)-y(t_{k+1}^-)) \Delta t  - g (x(t_{k+1}^-)- y(t_{k+1}^-)) \Delta B_k       \\
  =  \int_{t_k}^{t_{k+1}^-} \big [ f(x(t))  - f(x (t_k) - y(t_k)) \big] \mathrm{d}t  + \int_{t_k}^{t_{k+1}^-} \big [ g(x(t))  - g(x (t_k) - y(t_k)) \big] \mathrm{d} B(t) 
\end{multline*}
and, due to $y(0)=0$, 
\begin{multline}     \label{conv-y-tk}
     y(t_{k+1}) =  \sum_{j=0}^k \int_{t_j}^{t_{j+1}^-}  \big [ f(x(t))  - f(x (t_j) - y(t_j)) \big] \mathrm{d}t     + \sum_{j=0}^k \int_{t_j}^{t_{j+1}^-} \big [ g(x(t))  - g(x (t_j) - y(t_j)) \big] \mathrm{d} B(t)       \\ 
      =  \int_{0}^{t_{k+1}}  \big [ f(x(t))  - f(x ( [t] ) - y( [t] )) \big] \mathrm{d}t    +  \int_{0}^{t_{k+1}}  \big [ g(x(t))  - g(x ( [t] ) - y( [t] )) \big] \mathrm{d} B(t) ,   \quad
\end{multline}
where $[t] := \sup \{ t_j : t_j \le t, j \in \NN \}$ for all $t \ge 0$. By Cauchy-Schwaz inequality, this produces
\begin{multline*}    
   | y(t_{k+1}) |^2 =  \left|  \int_{0}^{t_{k+1}}  \big [ f(x(t))  - f(x ( [t] ) - y( [t] )) \big] \mathrm{d}t  +  \int_{0}^{t_{k+1}}  \big [ g(x(t))  - g(x ( [t] ) - y( [t] )) \big] \mathrm{d} B(t) \right|^2      \\
   \le 2 \left[ t_{k+1} \int_{0}^{t_{k+1}}  \big | f(x(t))  - f(x ( [t] ) - y( [t] )) \big|^2 \mathrm{d}t    + \left| \int_{0}^{t_{k+1}}  \big [ g(x(t))  - g(x ( [t] ) - y( [t] )) \big] \mathrm{d} B(t) \right|^2   \right] .   
\end{multline*}
By the It{\^o}  isometry and the global Lipschitz condition (\ref{globalLipschitz_app}),  
\begin{multline*}    
  \EE \, | y(t_{k+1}) |^2     \le 2   t_{k+1} \, \EE  \int_{0}^{t_{k+1}}  \big | f(x(t))  - f(x ( [t] ) - y( [t] )) \big|^2 \mathrm{d}t    \\
{} +   2   t_{k+1}   \,   \EE \int_{0}^{t_{k+1}}  \big | g(x(t))  - g(x ( [t] ) - y( [t] )) \big|^2 \mathrm{d} t \\
    \le 2  L^2  ( t_{k+1}  +1) \; \EE \int_{0}^{t_{k+1}}  \big | x(t)   -  x ( [t] ) + y( [t] ) \big|^2 \mathrm{d}t .   
\end{multline*}
Since (\ref{SDE_x_conv}) and (\ref{SDE_y_conv}) give $x(t)   -  x ( [t] ) = y(t) - y([t])$
for all $t \ge 0$, 
this implies
\begin{equation}     \label{conv-Ey2-tk-2}
   \EE \, | y(t_{k+1}) |^2  \le  2  L^2  ( t_{k+1}  +1) \; \EE\,  \int_{0}^{t_{k+1}}   \,  | y(t)   |^2 \mathrm{d}t.
\end{equation}
For any $T \ge 0$, using  the It{\^o} lemma, (\ref{globalLipschitz_app}),  (\ref{LV_1-y}) and (\ref{conv-Ey2-tk-2}), one obtains
\begin{multline*}   
  \EE \, | y(T) |^2    = \EE \, |y([T]) |^2   + \EE \, \int_{[T]}^T \big[ 2 y^T(s) f(x(s)) + | g(x(s) |^2 \big]  \mathrm{d} s     \\
   \le 2  L^2  ( [T]  +1) \; \EE  \, \int_{0}^{[T]}   | y(s)   |^2 \mathrm{d}s   + \EE \, \int_{[T]}^T  | y(s) |^2  \mathrm{d} s    +  \EE \, \int_{[T]}^T \big[ | f(x(s)) |^2+ | g(x(s) |^2 \big]  \mathrm{d} s    \\
    \le  K_T \; \EE  \, \int_{0}^{[T]}   | y(s)   |^2 \mathrm{d}s  + 2 L^2 \; \EE \, \int_{[T]}^T | x(s)|^2  \mathrm{d} s ,
\end{multline*}
where constant $K_T =  2  L^2  ( [T]  +1) \vee 1$. This implies 
\begin{multline*}  
   \EE  \Big[ \sup_{0\le t \le T}  | y(t) |^2 \Big]     \le 2 L^2    \int_{0}^{\Delta t} \EE    \Big[  \sup_{0 \le t_j \le [T]} | x(t_j +s) |^2 \Big]  \mathrm{d} s     +  K_T   \int_{0}^{T} \EE    \Big[  \sup_{0 \le s \le t} |y(s)|^2  \Big]    \mathrm{d} t       \\
  \le 2C_T L^2 {\Delta t} + K_T     \int_{0}^{T} \EE  \, \Big[  \sup_{0 \le s \le t} |y(s)|^2  \Big] \; \mathrm{d} t, 
\end{multline*}
where $C_T$ is  given by (\ref{conv-Esupx}) above.
In view of the Gronwall inequality  (\cite[Lemma 4.5.1, p129]{kloeden1992}, \cite[Theorem 8.1, p45]{mao2007book}), this yields
\begin{equation*} 
  \EE \,\left[ \sup_{0\le t \le T}  | y(t) |^2 \right]  \le   2C_T L^2  e^{K_T T} {\Delta t},
\end{equation*}
which completes the proof.
\end{proof}

%


\begin{thebibliography}{100}

\bibitem{arnold1974}
L.~Arnold,  \emph{Stochastic differential equations: theory and applications}, 
 New York, US: John Wiley \& Sons, 1974.

\bibitem{astrom1970}
K.~J. {\r A}str{\"o}m,  \emph{Introduction to stochastic control theory},  New York, US: Academic Press, 1970.


\bibitem{astrom1997}
K.~J. {\r A}str{\"o}m and B.~Wittenmark, \emph{Computer-controlled systems: theory and design (3rd Ed.)},  New Jersey, US: Prentice Hall, 1997.

\bibitem{astrom2014}
K.~J. {\r A}str{\"o}m and P.~R. Kumar,  
``Control: a perspective,"
  \emph{Automatica}, vol. 50,  pp. 3-43, 2014.



%
\bibitem{boyed1994}
S.~Boyed, L.~El Ghaoui, E.~Feron and V.~Balakrishnan,   
  \emph{Linear matrix inequalities in systems and control theory},
 Pennsylvania, US: Society for Industrial and Applied Mathematics, 1994. 
%
\bibitem{buckwar2012a}
E.~Buckwar and C.~Kelly, 
``A structural analysis of asymptotic mean-square stability
for multi-dimensional linear stochastic differential systems,"
  \emph{Comput.   Math.   Appl.}, vol. 64, pp. 2282-2293, 2012.

\bibitem{buckwar2012b}
E.~Buckwar and T.~Sickenberger,  
``Non-normal drift structures and linear stability analysis of numerical
methods for systems of stochastic differential equations,"
 \emph{Appl. Numer. Math.},  vol. 62, pp. 842-859, 2012.


\bibitem{caraballo2017}
T.~Caraballo, M.~A. Hammami and L.~Mchiri, 
``Practical exponential stability of impulsive stochastic functional differential equations,"
  \emph{Syst.   Contr. Lett.}, vol. 109,  pp. 43-48, 2017.

\bibitem{control2002}
Control Panel, 
``Control in an information rich world: report of the panel on future directions in
control, dynamics, and systems,"
 Available:  {\verb#http://www.cds.caltech.edu/~murray/#}
 {\verb#cdspanel/report/latest.pdf#}, accessed on April 11, 2022.


\bibitem{cps2008}
 CPS Summit,
``Report: Cyber-physical systems summit,"
 Available:   {\verb#http://iccps2012.cse.#}  {\verb#wustl.edu/_doc/CPS_Summit_Report.pdf#}, accessed on April 11, 2022.


%
\bibitem{elghaoui1995}
L.~El Ghaouii,  
``State-feedback control of systems with multiplicative noise via
linear matrix inequalities," \emph{Systems \& Control Letters}, vol. 24, pp. 223--228, 1995.


\bibitem{fridman2014}
 E.~Fridman, ``Tutorial on Lyapunov-based methods for time-delay systems,"
  \emph{Europ. J. Contr.}, vol.20,  pp.271-283, 2014.

%
\bibitem{gahinet1995}
P.~Gahinet, A.~Nemirovski, A.~J. Laub and M.~Chilali, \emph{LMI control toolbox}, 
  Massachusetts, USA: The MathWorks Inc, 1995.


\bibitem{gikhman1996}
I.~I. Gikhman and A.~V. Skorokhod,  \emph{Introduction to the theory of random processes},
 Philadelphia, USA: Dover Publications, 1996.

\bibitem{heirung2018}
T.~A. N. Heirung, J.~A. Paulson, J.~O'Leary, A.~Mesbah,
``Stochastic model predictive control - how does it work?"
  \emph{Comput.   Chem. Engin.}, vol. 114, pp. 158-170, 2018.


\bibitem{hespanha2008}
J.~P. Hespanha, D.~Liberzon and  A.~R. Teel, 
``Lyapunov conditions for input-to-state stability of impulsive systems,"
 \emph{Automatica}, vol. 44,  pp. 2735-2744, 2008.


\bibitem{higham2000}
D.~J. Higham, ``Mean-square and asymptotic stability of the stochastic theta method,"
  \emph{SIAM J. Numer. Anal.}, vol. 38,  pp. 753-769, 2000.




\bibitem{higham2001}
D.~J. Higham, ``An algorithmic introduction to numerical simulation of stochastic differential equations,"
  \emph{SIAM Review}, vol. 43,  pp. 525-546, 2001.


\bibitem{higham2002}
D.~J. Higham, X.~Mao and A.~M.Stuart, ``Strong convergence of Eular-type methods for nonlinear stochastic differential equations,"
  \emph{SIAM J. Numer. Anal.},  vol. 40,  pp. 1041-1063, 2002.


\bibitem{higham2003}
D.~J. Higham, X.~Mao and A.~M.  Stuart, ``Exponential mean-square stability of numerical solutions to  stochastic differential equations,"
  \emph{LMS J. Comput. Math.}, vol. 6,  pp. 297-313, 2003.


\bibitem{higham2006}
D.~J. Higham and X.~Mao, ``Nonnormality and stochastic differential equations,"
  \emph{BIT Numer. Anal.}, vol. 46,  pp. 525-532, 2006.



\bibitem{higham2008}
D.~J. Higham, ``Modeling and simulating chemical reactions,"
 \emph{SIAM Review}, vol. 50,  pp. 347-368, 2008.




\bibitem{huang2009}
L.~Huang and X.~Mao, ``On input-to-state stability of stochastic retarded systems with Markovian switching,"
  \emph{IEEE Trans.  Automat. Contr.}, vol. 54, pp. 1898-1902, 2009.


\bibitem{huang2010PhD}
L.~Huang, ``Stability and stabilisation of stochastic delay systems,"
 University of Strathclyde, Glasgow, UK, PhD thesis, 2010.


\bibitem{huang2010}
L.~Huang and X.~Mao, ``SMC design for $H_\infty$ control of uncertain stochastic delay systems,"
  \emph{Automatica}, vol. 46,  pp. 405-412, 2010.





\bibitem{huang2015}
L.~Huang, H.~Hjalmarsson and H.~Koeppl, ``Almost sure stability and stabilization of discrete-time stochastic systems,"
  \emph{Syst. Contr. Lett.}, vol. 82,  pp. 26-32, 2015. 

\bibitem{huang2016}
L.~Huang, L.~Pauleve, C.~Zechner, M.~Unger, A.~S. Hansen and H.~Koeppl, ``Reconstructing dynamic molecular states from single-cell time series," \emph{J. R. Soc. Interface}, vol. 13: 20160533, 2016. 


\bibitem{huang2022}
L.~Huang, S.~Xu, ``Impulsive stabilization of  systems with control delay,"
  \emph{IEEE Trans.  Automat. Contr.}, to appear. 

\bibitem{huang_partII}
L.~Huang, ``Stability of stochastic impulsive differential equations:  a theoretic foundation for cyber-physical systems", 
  to be submitted. 


\bibitem{hutzenthaler2012}
M.~Hutzenthaler, A.~Jentzen and P.~E. Kloedenl, ``Strong convergence of an explicit numerical
method for SDEs with non-globally Lipschitz continuous coefficients,"
  \emph{Ann. Appl. Probab.}, vol. 22,  pp. 1611-1641, 2012. 


\bibitem{khas2012}
R.~Khasminskii,   \emph{Stochastic stability of differential equations (2nd edition)},
  Berlin, Gernany: Springer-Verlag, 2012.


\bibitem{kloeden1992}
P.~E.  Kloeden and E.~Platen,  
 \emph{Numerical solution of stochastic differential equations}, Berlin, Germany: Spinger-Verlag, 1992. 

\bibitem{koutsoukos2014}
X.~Koutsoukos and P.~Antsaklis,  
``Design science for cyberphysical systems,"
in \emph{The Impact of Control Technology},  2nd ed., T.~Samad and A.~M. Annaswamy,  Ed.  IEEE Control Systems Society, 2014. Available: {\verb#http://ieeecss.org/sites/ieeecss/files#}  {\verb#/2019-06/IoCT2-RC-Koutsoukos-1.pdf#}, accessed on April 11, 2022.



\bibitem{lee2010}
E.~A. Lee, ``CPS Foundations," 
in \emph{Proc. of 47th IEEE/ACM Design Automat. Conf.}, Chicago, IL, USA, 2010.


\bibitem{lee2015}
E.~A. Lee, ``The past, present and future of cyber-physical systems: a focus on models,"
 \emph{Sensors}, vol. 15,  pp. 4837-4869, 2015.


\bibitem{lee2017}
E.~A. Lee and  S.~A. Seshia,   \emph{Introduction to embedded systems - a cyber-physical systems approach (2nd edition)},
  Massachusetts, US: The M.I.T. Press, 2017.


\bibitem{liu1988}
Y.~Liu, Z.~Song,  \emph{Theory and application of large-scale dynamic systems: decomposition, stability and structure},  Guangzhou, China: South China Univ.  Technol. Press, 1988 (in Chinese).

\bibitem{liu1992}
Y.~Liu and Z.~Feng,  \emph{Theory and application of large-scale dynamic systems: stochastic stability and control},  Guangzhou, China: South China Univ.  Technol. Press, 1992 (in Chinese).






\bibitem{mao2007book}
X.~Mao,  \emph{Stochastic differential equations and applications (2nd edition)}, Chichester, UK: Horwood Publishing, 2007. 


\bibitem{mao2015a}
X.~Mao, ``Almost sure exponential stability in the numerical simulation of stochastic differential equations,"
  \emph{SIAM J. Numer. Anal.}, vol. 53,  pp.370-389, 2015.

\bibitem{mao2015b}
X.~Mao, ``The truncated Euler-Maruyama method for stochastic differential equations,"
  \emph{J. Comput. Appl. Math.}, vol. 290,  pp.370-384, 2015.


\bibitem{nagh2008}
 P.~Naghshtabrizi,  J.~P. Hespanha, A.~R. Teel, ``Exponential stability of impulsive systems with application to uncertain
sampled-data systems,"
  \emph{Syst.   Contr. Lett.}, vol.57,  pp.378-385, 2008.

\bibitem{nghiem2006}
T.~Nghiem, G.~J. Pappas, A.~Girard and R.~Alur, ``Time-triggered implementations of dynamic controllers," in \emph{Proc. of 2013 Intern. Conf.   Embedded Software}, Seoul, Korea, 2006.


\bibitem{oksendal2000}
B.~$\varnothing$ksendal,   \emph{Stochastic differential equations: an introduction with applications (5th edition)}, 
 New York, USA: Springer, 2000. 


\bibitem{peng2010}
S.~Peng and Y.~Zhang, ``Razumikhin-type theorems on $p$th moment
exponential stability of impulsive stochastic
delay differential equations," 
  \emph{IEEE Trans.  Automat. Contr.}, vol. 55, pp. 1917-1922, 2010.

\bibitem{sabanis2013}
S.~Sabanis, ``A note on tamed Euler approximations,"
 \emph{Electron. Commun. Probab.},  vol. 47, pp. 1-10, 2013.

\bibitem{saito1996}
Y.~Saito and T.~Mitsui, ``Stability analysis of numerical schemes for stochastic differential equations,"
  \emph{SIAM J. Numer. Anal.}, vol. 33,  pp. 2254-2267, 1996.

\bibitem{samoilenko1995}
A.~M. Samoilenko and N.~ A. Perestyuk,  
 \emph{Impulsive differential equations},
  Singapore: Word Scientific, 1995.

\bibitem{sarkka2019}
S.~S{\"a}rkk{\"a} and A.~ Solin,  
 \emph{Applied stochastic differential equations},
  Cambridge, UK: Cambridge University Press, 2019.



 
\bibitem{scotton2013}
F.~Scotton, L.~Huang, S.~A. Ahmadi and B.~Wahlberg,  ``Physics-based Modeling and Identification for HVAC Systems," 
in \emph{Proc. of 2013 Europ. Contr. Conf.}, Zurich, Switzerland, 2013.



\bibitem{stuart1996}
A.~M. Stuart and P.~Humphries,   \emph{Dynamical systems and numerical analysis},
  Cambridge, UK: Cambridge Univ. Press, 1996.





\bibitem{tsien2015}
H.~S. Tsien,  \emph{Engineering Cybernetics},
 Shanghai, China: Shanghai Jiaotong Univ.   Press, 2015. 


\bibitem{vaidya2015}
U.~Vaidya, ``Stochastic stability analysis of discrete-time system using Lyapunov measure,"
  \emph{Americ. Contr. Conf.}, Chicago, IL, USA, 2015.


\bibitem{vankampen2007}
N.~G. Van Kampen,  \emph{Stochastic processes in physics and chemistry (3rd edition)},
  Amsterdam, Netherlands: Elsevier, 2007.


\bibitem{wiener1961}
N.~Wiener,   \emph{Cybernetics or control and communication in the animal and the machine (2nd edition)}, 
 Massachusetts, US: The M.I.T. Press, 1961.


\bibitem{wilkinson2012}
D.~Wilkinson,  \emph{Stochastic modelling for systems biology (2nd edition)},  London, UK: CRC Press, 2012.



 

\bibitem{yang2001}
T.~Yang,   \emph{Impulsive control theory},
  Berlin, Gernany: Springer-Verlag, 2001.

 
 




\end{thebibliography}
\end{document}